\documentclass[10pt,twocolumn,letterpaper]{article}

\usepackage{cvpr}              

\usepackage[dvipsnames]{xcolor}

\definecolor{cvprblue}{rgb}{0.21,0.49,0.74}
\usepackage[pagebackref,breaklinks,colorlinks,citecolor=cvprblue]{hyperref}

\usepackage{mathtools}

\usepackage{amsmath,amsfonts,bm}

\def\eqref#1{equation~\ref{#1}}

\def\1{\bm{1}}

\def\mA{{\bm{A}}}
\def\mB{{\bm{B}}}

\def\mH{{\bm{H}}}
\def\mI{{\bm{I}}}

\def\mK{{\bm{K}}}

\def\mP{{\bm{P}}}
\def\mQ{{\bm{Q}}}

\def\mS{{\bm{S}}}

\def\mV{{\bm{V}}}
\def\mW{{\bm{W}}}
\def\mX{{\bm{X}}}
\def\mY{{\bm{Y}}}
\def\mZ{{\bm{Z}}}

\DeclareMathAlphabet{\mathsfit}{\encodingdefault}{\sfdefault}{m}{sl}
\SetMathAlphabet{\mathsfit}{bold}{\encodingdefault}{\sfdefault}{bx}{n}

\def\sD{{\mathbb{D}}}

\usepackage{amsmath,amssymb,amsfonts}
\usepackage{amsthm}
\usepackage{wrapfig}
\usepackage{url}
\usepackage{graphicx}
\usepackage{threeparttable} 
\usepackage{slashbox}
\usepackage{algorithm}
\usepackage{algorithmic}
\usepackage{colortbl}

\usepackage[accsupp]{axessibility}
\makeatletter
\robustify\@latex@warning@no@line
\makeatother
\usepackage{authblk}

\theoremstyle{plain}
\newtheorem{theorem}{Theorem}[section]

\newtheorem{lemma}[theorem]{Lemma}
\newtheorem{corollary}[theorem]{Corollary}
\theoremstyle{definition}

\theoremstyle{remark}

\newcommand{\blue}[1]{\textcolor{black}{#1}}

\definecolor{Gray}{gray}{0.8}

\def\paperID{14173} 
\def\confName{CVPR}
\def\confYear{2024}

\title{Permutation Equivariance of Transformers and Its Applications}

\author[$\dag$]{Hengyuan Xu}
\author[$\dag$]{Liyao Xiang\thanks{Corresponding author: xiangliyao08@sjtu.edu.cn. The research was supported in part by National Science and Technology Major Project 2021ZD0112801, NSF China (62272306, 62032020, 62136006).}}
\author[$\dag$]{Hangyu Ye}
\author[$\ddag$]{Dixi Yao}
\author[$\dag$]{Pengzhi Chu}
\author[$\ddag$]{Baochun Li}
\affil[$\dag$]{Shanghai Jiao Tong University}
\affil[$\ddag$]{University of Toronto}

\begin{document}
\maketitle
\begin{abstract}
	Revolutionizing the field of deep learning, Transformer-based models have achieved remarkable performance in many tasks. Recent research has recognized these models are robust to shuffling but are limited to inter-token permutation in the forward propagation. In this work, we propose our definition of permutation equivariance, a broader concept covering both inter- and intra- token permutation in the forward and backward propagation of neural networks. We rigorously proved that such permutation equivariance property can be satisfied on most vanilla Transformer-based models with almost no adaptation. We examine the property over a range of state-of-the-art models including ViT, Bert, GPT, and others, with experimental validations. Further, as a proof-of-concept, we explore how real-world applications including privacy-enhancing split learning, and model authorization, could exploit the permutation equivariance property, which implicates wider, intriguing application scenarios. The code is available at \url{https://github.com/Doby-Xu/ST}
\end{abstract}    


\section{Introduction}
Originating as a tool in natural language processing, Transformer now has permeated various fields such as computer vision, multi-modal tasks, etc. Transformer model, introduced in \cite{Transformer}, has revolutionized the way of approaching sequence-based tasks, and further demonstrates versatility and power in a diverse set of tasks with the development of Bert\cite{Bert}, GPT \cite{GPT2,GPT3}, ViT \cite{ViT}, etc. Meanwhile, its intriguing property is discovered, e.g., the outputs being robust or invariant to token shuffling \cite{naseer2021intriguing,lee2019set}. However, previous works either show the property through empirical observation \cite{naseer2021intriguing}, or by modifying the original model structure \cite{lee2019set}, typically for input-order-independent tasks.

The shuffling invariance property of Transformer is widely recognized but not clearly understood by the community. ViT shows superior permutation robustness compared to CNN in patch shuffling where the spatial information is totally disrupted \cite{naseer2021intriguing}, indicating that self-attention is invariant to patch ordering. 
Similarly, 
\cite{lee2019set} proposes Set Transformer, satisfying \textit{permutation equi-/in-variance} --- the output of the model should not change under any permutation of the elements in the input set. Their model inherits the network architectures proposed by \cite{zaheer2017deep} where each input element is first independently fed into a feed-forward neural network and then aggregated by a pooling operation.

As we found, the previous notion of \textit{permutation equi-/in-variance} is quite limited. In this work, we propose our definition of \textit{permutation equivariance} --- meaning that the model trained over any inter- or intra-token permutation is equivalent to the model trained over normal inputs, in contrast to the inter-token permutation in previous works. Our definition essentially pushes one step forward in two-folds: first, the former definition works at a coarser granularity, i.e., tokens are exchanged and permuted while ours incorporates both inter- and intra-token shuffling. Second, our proposed property is stricter in the sense that, it not only requires output equivariance under input permutation in forward propagation (forwarding equivariance) but also demands model weights to be equivalently 
trained in the backward propagation (backprop equi-/in-variance). We show such backprop equi-/in-variance is closely related to forwarding equivariance by modeling the permutation as row/column shuffling in matrices. 


Our findings inspire many potential applications in real-world scenarios. We present a simple, yet effective privacy-enhancing technique for feature offloading, based on the permutation equivariance property. An honest-but-curious cloud merely provides computational service while being blind to the training data, inference data, and the trained model. We show that our method remarkably improves the data utility-privacy tradeoff which is intensively investigated in privacy-preserving split learning \cite{patchshuffling,jeong2022privacy,singh2021disco}. Moreover, the trained model is protected from being effectively fine-tuned by unauthorized parties, yet with almost no impact on the model's performance.

Highlights of our contributions include: first, we reveal the permutation equivariance property in the forward and backward propagation of neural networks. Second, by analyzing the inner workings of the Transformer model, we prove that a wide range of models satisfy permutation equivariance. Third, as a proof-of-concept, we design a privacy-enhancing mechanism and a model authorization scheme to show the promising use of the property. A series of experiments demonstrate the superiority of our design to the state-of-the-art methods in the application scenes.

\section{Related Works}
\label{sec:relatedworks}

\subsection{Transformer}
Dispensing with recurrence and convolutions, Transformer \cite{Transformer} shows superior performance solely on attention mechanisms in translation tasks, and soon becomes the de-facto standard for natural language processing (NLP). Models such as Bert \cite{Bert}, GPT \cite{GPT2, GPT3}, and a myriad of Transformer-based counterparts \cite{lewis2019bart, liu2019roberta} have consistently achieved state-of-the-art results across a wide spectrum of tasks. Most recently, with Transformer models as the core, Large Language Models (LLMs) have achieved great success in recognizing, translating, predicting, or generating text or other contents.

Inspired by its success in NLP, Vision Transformer (ViT \cite{ViT}) was designed and its performance exceeds the state-of-the-art convolutional networks on a variety of image tasks such as classification \cite{touvron2022deit}, object detection \cite{ViTDet}, and semantic segmentation \cite{ViTAdapter}. Substantial efforts have also been devoted into designing powerful pre-trained models \cite{DeiT, BEiT, Dino} with ViT-based backbone, leading to the emergence of a growing family of networks \cite{bao2022all, carion2020end, ViTAdapter, ViTDet}.

\subsection{Permutation Equi-/In-variance}
Permutation invariance property refers to that the output value for a given set is the same regardless of the order of objects in the set \cite{zaheer2017deep,lee2019set}. A closely related property permutation equivariance describes that for function $f$, any permutation $\mP$ and input $\mX$, $f(\mP \mX) = \mP f(\mX)$. Commonly, the permutation takes place at the token level and poses additional requirements on models. For example, Deep sets \cite{zaheer2017deep} derives the necessary and sufficient conditions for permutation invariance in deep models, i.e., the function can be decomposed in the form $\rho(\sum_{\mX \in\mathcal{X} } \phi(\mX))$ for transformations $\phi$ and $\rho$. Set transformer \cite{lee2019set} further instantiates $\phi$ and $\rho$ to an adapted encoder and a decoder, respectively. But 
their permutation invariance limits to the neural network forwarding (or inference), but neglects the invariance in the backward propagation (or training). In contrast, our work does not alter the existing Transformer network structure and illustrates permutation properties in both forward and backward propagations.

The (forwarding) permutation equivariance property has been widely recognized and applied in many works \cite{naseer2021intriguing, patchshuffling, engel2021point, tang2021sensory, lu2023pinat}. It is empirically verified by \cite{naseer2021intriguing} that ViT is more robust to patch shuffling compared to convolutional networks. Leveraging such a property, Yao \emph{et al.} propose a privacy-preserving split learning framework \cite{patchshuffling} by patch-shuffling  the embedding. The permutation invariance property also sees applications in point cloud processing \cite{engel2021point}, reinforcement learning \cite{tang2021sensory}, neural architecture search \cite{lu2023pinat}, etc. However, most target tasks are input-order-insensitive (set-input problem). Different from them, our applications include tasks that rely on the order of the input (e.g., generating text).

\begin{table}[]
	\caption{The notations used.}
	\label{tab:notation}
	\begin{tabular}{lll}
		\hline
		Symbol         & Description                                                                & Shape                     \\ \hline
		$\mX, \mY$     & Data, label                                                             & Task dependent   \\
		$F_1, F_2$     & Embedding, task head        & Task dependent   \\
		$\mZ$          & Feature/ output of $F_1$ & $\mathbb{R}^{n \times d}$ \\
		$\hat{\mY}$    & Prediction/ output of $F_2$ & Same shape w/ $\mY$\\
		$\mathrm{Enc}(\cdot)$ & Transformer encoder                                                        & $\mathbb{R}^{n \times d} \to \mathbb{R}^{n \times d}$ \\ 
		$\mathrm{T}(\cdot)$   &Transformer backbone                                  & $\mathbb{R}^{n \times d} \to \mathbb{R}^{n \times d}$ \\
		$\mP_R$        & Row permutation matrix                                                   & $\mathbb{R}^{n \times n}$ \\
		$\mP_C$        & Column permutation matrix                                                  & $\mathbb{R}^{d \times d}$ \\
		$\mW$          & Weights of Transformer & $\mathbb{R}^{t \times d}$   \\ \hline
	\end{tabular}
\end{table}

\section{Notations}
\label{sec:notations}
\begin{figure}
    \centering
    \includegraphics[width=0.49\textwidth]{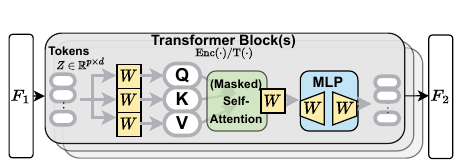}
    \caption{Illustration of Transformer backbone. Learnable weights in permutation are expressed by yellow blocks.}
	\label{fig:transformer}
\end{figure}

\begin{figure*}[t]
	\centering
	\includegraphics[width=0.96\linewidth]{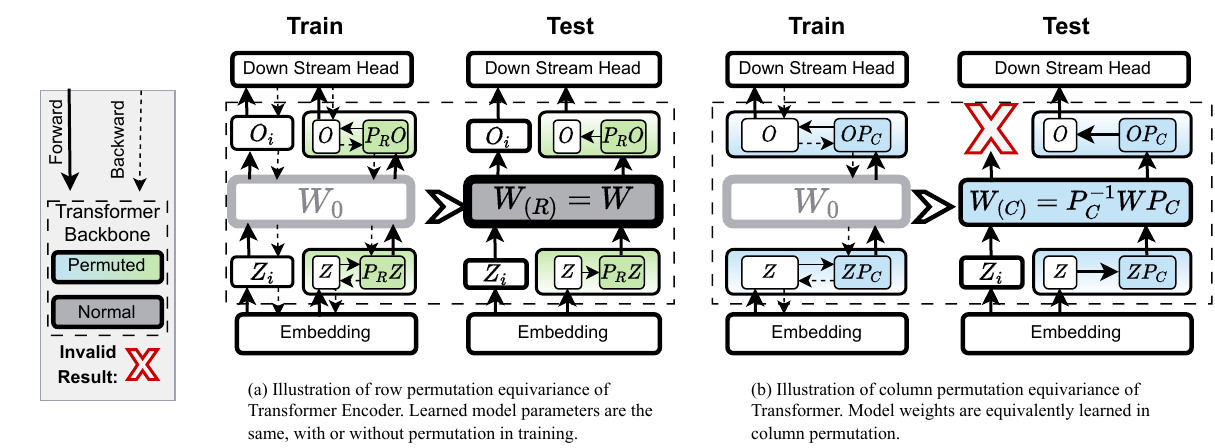}
	\caption{Illustration of permutation properties. $\mW$ indicates main parameters in Transformer backbone (stacked Transformer encoders and decoders).}
	\label{fig:fig1}
\end{figure*}
We summarize the notations used in this paper in Tab.~\ref{tab:notation}, where $n$ is the number of tokens and $d$ is the dimension of tokens. We use $\mW \in \mathcal{R}^{t \times d}$ to generally denote the weight matrices of the Transformer backbone including weights in QKV projection, attention projection, and MLP. Details about the location of $\mW$ in the structure of Transformer backbone are shown in Fig.~\ref{fig:transformer}. Transformer backbone usually consists of stacked encoders and decoders. Some components such as bias and LayerNorm do not have matrix-shaped weights, so we leave their discussion to Sec.~\ref{sec:gamma_bias_multihead_attention}. 

The matrix form of weights allows us to perform row or column permutations. We employ the subscript $_{(R)}$ to refer to row permutation of the corresponding features, weights, derivatives, etc., while using subscript $_{(C)}$ for column permutations. The subscript of $_{(P)}$ generally refers to row, column, and both row and column permutations. Specifically, the row permutation is equivalent to the inter-token permutation in previous works and column permutation is essentially intra-token shuffling. Row/column permutation is expressed by matrix multiplication with the permutation matrix: $\mP_R \mZ$, $\mZ \mP_C$, or $\mP_R\mZ \mP_C$. The unshuffling is represented by multiplying the inverse of the permutation matrix, i.e., $\mP_R^{-1}\mZ_{(P)}$, $\mZ_{(P)}\mP_C^{-1}$, or $\mP_R^{-1}\mZ_{(P)}\mP_C^{-1}$.

\section{Properties and Proofs}
We illustrate the key properties and provide their proofs in this section. Due to space limit, we collect most proofs in Appendix~\ref{sec:permutation_equivariant_operators} and~\ref{sec:proofEnc}.

\subsection{Permutation Equivariance of Transformer}
\label{sec:analysis}
In accordance with the inter-token permutation equivariance of the Transformer encoder \cite{lee2019set}, the row permutation equivariance can be expressed as
\begin{theorem}[\textbf{Row Permutation Forward Equivariance}]
    \label{thm:token_permutation_equivariance}
    Transformer encoder is permutation equivariant w.r.t. token permutations, i.e. the row permutation of the input matrix, in forward propagation, i.e., $\mathrm{Enc}(\mP_R \mZ) = \mP_R \mathrm{Enc} (\mZ)$ for any permutation matrix $\mP_R \in \mathbb{R}^{n \times n}$.
\end{theorem}

Following the property, we perform row permutation to the output of the embedding layer $F_1$ and invert the permutation at the input to the downstream task head $F_2$, i.e., for 
\begin{equation}
    \label{eq:training_scenario}
    \mZ = F_1(\mX), \hat{\mY}= F_2(\mathrm{Enc}(\mZ)),
\end{equation}
we perform shuffling and unshuffling as
\begin{equation}
    \label{eq:training_scenario_permuted}
    \mZ_{(R)} = \mP_R F_1(\mX), \hat{\mY}_{(R)} = F_2(\mP_R^{-1}\mathrm{Enc}(\mZ_{(R)})).
\end{equation}
According to Thm.~\ref{thm:token_permutation_equivariance}, the permuted input to $F_2$ is $\mP_R^{-1}\mathrm{Enc}(\mP_R\mZ) = \mathrm{Enc}(\mZ)$. Given the shuffling and unshuffling operations, we can derive the following key property:
\begin{theorem}[\textbf{Row Permutation Backward Invariance}]
    \label {thm:gradient_permutation_backward_Invariance}
    In backward propagation, gradients of the Transformer encoder in the natural setting (Eq.~\ref{eq:training_scenario}) and the permuted setting (Eq.~\ref{eq:training_scenario_permuted}) are the same: 
    \begin{equation}
        \label{eq:gradient_permutation_Invariance}
        \frac{\partial l} {\partial \mW} =\frac{\partial l} {\partial \mW_{(R)}},
    \end{equation}
    given loss function $l$.
\end{theorem}

 With simple induction, we have:
\begin{corollary}
    \label{cor:learning_permutation_backward_Invariance}
    The weights of the Transformer encoder learned in the natural setting (Eq.~\ref{eq:training_scenario}) and that learned in the permuted setting (Eq.~\ref{eq:training_scenario_permuted}) are the same: 
    \begin{equation}
        \label{eq:learning_backward_Invariance}
        \mW = \mW_{(R)}.
    \end{equation}
\end{corollary}
The shuffling and inverse shuffling procedures and the properties are displayed in Fig.~\ref{fig:fig1}(a). In training, token embeddings can be randomly shuffled without affecting the learning results, provided that the input to the task head is inversely permuted. `Without affecting' here means that the trained weights are exactly the same as the learned weights without any permutation.

Similar to row permutation, we could also perform column shuffling to the input, i.e., intra-token permutation:
\begin{equation}
    \label{eq:training_scenario_col_permuted}
    \mZ_{(C)} = F_1(\mX)\mP_C, \hat{\mY}_{(C)} = F_2(\mathrm{T}_{(C)}(\mZ_{(C)}) \mP_{C}^{-1}). 
\end{equation}
Unfortunately, permutation invariance does not hold column-wise, i.e., $\mathrm{T}(\mZ\mP_C) \neq \mathrm{T}(\mZ)\mP_C$; rather, we could achieve $\mathrm{T}_{(C)}(\mZ\mP_C) = \mathrm{T}(\mZ)\mP_C$ by adjusting the weights in $\mathrm{T}$ according to
\begin{equation}
    \label{eq:weight_permutation}
    \mW_{(C)} = \mP_C^{-1}\mW\mP_C.
\end{equation}
Formally, the property holds as
\begin{theorem}[\textbf{Column Permutation Forward Equivariance}]
    \label{thm:column_permutation_equivariance}
    Stacked Transformer encoder and decoder is permutation equivariant w.r.t. column permutations, in forward propagation, i.e., $\mathrm{T}_{(C)}(\mZ\mP_C) = \mathrm{T}(\mZ)\mP_C$ for any permutation matrix $\mP_C \in \mathbb{R}^{d \times d}$, where weights in $\mathrm{T}_{(C)}$ are weights in $\mathrm{T}$ permuted by Eq.~\ref{eq:weight_permutation}.
\end{theorem}

With Thm.~\ref{thm:column_permutation_equivariance}, we can easily derive that $\hat{\mY}_{(C)}$ in Eq.~\ref{eq:training_scenario_col_permuted} is equal to $\hat{\mY}$ in Eq.~\ref{eq:training_scenario}. This property is depicted in Fig.~\ref{fig:fig1}(b). The case bears resemblance to homomorphic encryption, wherein the input $\mZ$ and backbone weights $\mW$ are `encrypted' by $\mP_C$, and the output is `decrypted' using $\mP_C^{-1}$. Although shuffling does not offer the same level of security compared to homomorphic encryption, it disguises the original inputs and models to some extent. Only those who possess the key $\mP_C$ can obtain the correct output and effectively use the model. Without the key, however, it is hard to guess the correct output, or take advantage of the model (see Sec.~\ref{sec:experiments}).

Likewise, in the corresponding backprop of column-permuted forwarding (Eq.~\ref{eq:training_scenario_col_permuted}), the property holds as:
\begin{theorem}[\textbf{Column Permutation Backward Equivariance}]
    \label{thm:column_permutation_backward_equivariance}
    In backward propagation, gradients of Transformer parameters in the natural setting (Eq.~\ref{eq:training_scenario}) and the permuted setting (Eq.~\ref{eq:training_scenario_col_permuted}) are correlated by the following equation:
    \begin{equation}
        \label{eq:column_permutation_backward_equivariance}
        \frac{\partial l} {\partial \mW_{(C)}}=\mP_C^{-1}\frac{\partial l} {\partial \mW}\mP_C
    \end{equation}
    where $l$ is the loss function.
\end{theorem}

With simple induction, we have:
\begin{corollary}
    \label{cor:column_permutation_weights_equivariance}
    The weights of Transformer learned in the natural setting (Eq.~\ref{eq:training_scenario}), $\mathrm{T}$, and that learned in the permuted setting (Eq.~\ref{eq:training_scenario_col_permuted}), $\mathrm{T}_{(C)}$, are correlated by $\mW_{(C)} = \mP_C^{-1}\mW\mP_C$, given the same randomly initialized weights.
\end{corollary}

\subsection{General Permutation Equivariant Networks}
\label{sec:general_permutation_equivariance_network}
In this section, we show that the row and column permutation equivariance in the forward and backward propagations can be generalized to a broader class of permutation equivariant networks, apart from Transformers. Here permutation generally refers to row and column permutation:
\begin{equation}
    \label{eq:training_scenario_compose}
    \mZ_{(P)} = \mP_R F_1(\mX) \mP_C,~~\hat{\mY}_{(P)} = F_2(\mP_R^{-1}f_{(P)}(\mZ_{(P)})\mP_C^{-1}),
\end{equation} 
where $f_{(P)}$ is the Transformer backbone $f$ has its weight permuted by Eq.~\ref{eq:weight_permutation}.

We show the conclusion in three steps: 1) a majority of common neural network operators satisfy forward permutation equivariance; 2) if an operator is permutation-equivariant in forwarding, it must be permutation-equivariant in the backward propagation; 3) a network composed by the aforementioned operators are both forward and backward permutation-equivariant.  

We provide proofs for 1) in Appendix~\ref{sec:permutation_equivariant_operators} covering a wide variety of operators, such as linear projection, attention, norms, element-wise operators (shortcut skip, Hadamard product, activation, etc.), and softmax. The exception is the type of operators working on sliding windows, e.g., convolutional. For 2) and 3), we will prove the following theorem:
\begin{theorem}[General Permutation Equivalent Networks]
    \label{thm:general_permutation_equivariance}
    If $f$ is composed by $f_N \circ f_{N-1} \circ \cdots \circ f_1$ where $f_1, \ldots, f_N$ that are permutation-equivariant in the forward propagation, and each contains weights (if any) as linear arguments, i.e., $f = f_N( \cdots f_2 ( f_1(\mZ \mW_1^{\top}) \mW_2^{\top}) \cdot \cdots \mW_N^{\top})$, $f$ is permutation-equivariant in the backward propagation, i.e., the weights $\mW$ in $f$ are associated with those of $f_{(P)}$ by $\mW_{(P)} = \mP_C^{-1}\mW\mP_C$. 
\end{theorem}

To prove Thm.~\ref{thm:general_permutation_equivariance}, we first prove the following lemma:
\begin{lemma}
    \label{lem:gradient_feature_equivalent}
    If $f$ is composed by permutation-equivariant operators as in Thm.~\ref{thm:general_permutation_equivariance}, the derivative of $f$ with respect to feature $\mZ$ in the natural (Eq.~\ref{eq:training_scenario}) and in the permuted (Eq.~\ref{eq:training_scenario_compose}) settings are correlated by
    \begin{equation}\label{eq:backpereq}
        \frac{\partial l} {\partial \mZ_{(P)}}=\mP_R\frac{\partial l} {\partial \mZ}\mP_C.
    \end{equation}
\end{lemma}

\begin{proof}
    It is obvious that in forwarding, the combination of permutation-equivariant operators remains to be permutation-equivariant since one can perform induction on the simple case of $ f_{1(P)}(f_{2(P)}(\mP_R \mZ \mP_C)) = f_{1(P)}(\mP_Rf_{2}( \mZ )\mP_C) =\mP_Rf_{1}(f_{2}( \mZ ))\mP_C $. Hence through all layers in $f$, it holds that
    \begin{equation}\label{eq:f2forward}
        f_{(P)}(\mZ_{(P)}) = \mP_R f(\mZ) \mP_C. 
    \end{equation}
    By Eq.~\ref{eq:training_scenario_compose} and Eq.~\ref{eq:f2forward}, the forward propagation feeds $f(\mZ)$ into $F_2$, which is equivalent to the vanilla forwarding, and hence the losses are the same.
    
    Now consider the backward propagation. As the loss $l$ is invariant by permutation, we differentiate $l$ through $\mZ_{i(P)}$ which is the intermediate output of layer $i$ of $f_{(P)}$. By the forward permutation equivariance, we know that it is associated with $\mZ_{i}$, the output of layer $i$ at $f$, as $\mZ_{i(P)} = \mP_R \mZ_{i} \mP_C$. Hence we have $\mathrm{d}\mZ_{i(P)} = \mP_R \mathrm{d}  \mZ_{i} \mP_C$. Therefore, the derivative of $l$ with respect to $\mZ_{i(P)}$ for any layer $i$ in the permuted setting is 
    \begin{align*}
        \mathrm{d}l&\triangleq \mathrm{tr}(\frac{\partial l}{\partial \mZ_{i(P)}} ^{\top} \mathrm{d}\mZ_{i(P)}) 
        =\mathrm{tr}(\frac{\partial l}{\partial \mZ_{i(P)}} ^{\top} \mP_R\mathrm{d}\mZ_{i}\mP_C)\\
        &=\mathrm{tr}(\mP_C\frac{\partial l}{\partial \mZ_{i(P)}} ^{\top} \mP_R\mathrm{d}\mZ_{i})
        =\mathrm{tr}((\mP_R^{\top}\frac{\partial l}{\partial \mZ_{i(P)}} \mP_C^{\top})^{\top} \mathrm{d}\mZ_{i}).
    \end{align*}
    The last equality suggests $\frac{\partial l}{\partial \mZ_{i}} =\mP_R^{\top}\frac{\partial l}{\partial \mZ_{i(P)}} \mP_C^{\top}$ according to Thm. 6 of \cite{hu2012matrix}. Since $\mP_R^{\top} = \mP_R^{-1}$ and $\mP_C^{\top} = \mP_C^{-1}$, Eq.~\ref{eq:backpereq} holds completing the proof.
\end{proof}

Then we prove Thm.~\ref{thm:general_permutation_equivariance}.
\begin{proof}
Without causing any confusion, let $\mZ_{i}$ denote the output of $f_i(f_{i-1}( \cdots f_1(\mZ \mW_1^{\top}) \cdots \mW_{i-1}^{\top})\mW_{i}^{\top})$. We differentiate the loss through $\mZ_{i}$ by
    \begin{align*}
        \mathrm{d}l &  \triangleq \mathrm{tr}((\frac{\partial l} {\partial \mZ_{i}})^{\top}\mathrm{d}\mZ_{i}) =  \mathrm{tr}((\frac{\partial l} {\partial \mZ_{i}})^{\top}\mathrm{d}( \mZ_{i-1} \cdot \mW_{i}^{\top}))\\
        & =   \mathrm{tr}((\frac{\partial l} {\partial \mZ_{i}})^{\top} \mZ_{i-1} \mathrm{d}\mW_{i}^{\top}) =  \mathrm{tr}((\frac{\partial l} {\partial \mZ_{i}}^{\top}\mZ_{i-1})^{\top}\mathrm{d}\mW_{i}).
    \end{align*}
    Hence the gradient of $\mW_{i}$ is $\frac{\partial l} {\partial \mW_i}=\frac{\partial l} {\partial \mZ_i}^{\top} \mZ_{i-1}$. Similarly, the gradient of $\mW_{i(P)}$ is as follows: 
    \begin{equation}
        \label{eq:general_proof_last}
        \begin{split}
            \frac{\partial l} {\partial \mW_{i(P)}}
            &=\frac{\partial l} {\partial \mZ_{i(P)}}^{\top}\mZ_{i-1(P)}=\mP_C^{\top}\frac{\partial l}{\partial \mZ_{i}}^{\top} \mP_R^{\top} \mP_R \mZ_{i-1}\mP_C\\
            &=\mP_C^{\top}\frac{\partial l} {\partial \mZ_{i}}^{\top}\mZ_{i-1}\mP_C=\mP_C^{-1}\frac{\partial l} {\partial \mW_{i}}\mP_C.
        \end{split}
    \end{equation}
The second equality holds due to permutation equivalence of $\mZ_{i-1}$ and Lem.~\ref{lem:gradient_feature_equivalent}. Following a similar argument to Corollary~\ref{cor:column_permutation_weights_equivariance}, we have $\mW_{i(P)} = \mP_C^{-1}\mW_{i}\mP_C$ for any layer $i$, given the same randomly initialized weights in the natural and permuted settings. Thereby the backward permutation equivariance holds. Proof completes.   
\end{proof}

It should be noted that the pair of arguments --- permutation forward and backward equivariance --- seem to be circular, i.e., the backward equivariance holds on the condition of the forward equivariance while the latter holds at $\mW_{(P)} = \mP_C^{-1}\mW\mP_C$. But one can break such a circle by permuting (the inputs and) the initial weights by Eq.~\ref{eq:weight_permutation} in the first place. Then the forward equivariance holds and one can train the model with Eq.~\ref{eq:training_scenario_col_permuted}. Eventually, the backward equivariance can be derived. In the case of training from scratch, it does not matter if the initial weights are permuted since they are all random.

\textbf{This shows why the permutation-invariance in Set Transformer \cite{lee2019set} is a pseudo-invariance.} The proposed permutation-invariance head, PMA \cite{lee2019set},  only achieves forward permutation invariance. However, input of PMA does not follow Lem.~\ref{lem:gradient_feature_equivalent}, but rather $\frac{\partial l} {\partial \mZ_{(P)}}=\frac{\partial l} {\partial \mZ}$ which destroys the permutation equivariant structure of Eq.~\ref{eq:general_proof_last}: instead of getting $\frac{\partial l} {\partial \mW_{i(P)}}=\frac{\partial l}{\partial \mZ_{i}}^{\top} \mP_R^{\top} \mP_R \mZ_{i-1}$, they obtain $\frac{\partial l} {\partial \mW_{i(P)}}=\frac{\partial l}{\partial \mZ_{i}}^{\top} \mP_R \mZ_{i-1}$ that the gradients are not aligned anymore.

\subsection{Other Components}
\label{sec:gamma_bias_multihead_attention}
Notably, the weights in Layer Norm ($\gamma$), bias ($b$) and MLP defined in the form of $\mathrm{MLP}(\mX) = \sigma(\mX\mW_1^\top)\mW_2^\top$ do not follow Eq.~\ref{eq:weight_permutation} for permutation equivariance. But the intuition is the same: pure token shuffling (row permutation) would not affect these operators; column permutation should be corrected by parameters permutation: 
\begin{align}
&   \gamma_{(C)} = \gamma\mP_C,~b_{(C)} = b \mP_C, \\
    &   \mW_{1(C)} = \mW_1\mP_C,~\mW_{2(C)} = \mP_C^{-1}\mW_2.
\end{align}

The proof can be found in Appendix~\ref{sec:permutation_equivariant_operators} and \ref{sec:proofEnc}. Parameters in $F_1$, including position embedding, are also not influenced. Proofs can be found in Appendix~\ref{proof_edge}.


\blue{It is worth noting that the position embedding in $F_1(\cdot)$ is not affected by permutation in both forward and back-propagation since it is added before permutation. The proof in Appendix \ref{proof_edge} shows that parameters outside of the permutation-inversion scheme are not affected.}

\section{Experiments}
\label{sec:experiments}
In this section, we provide empirical evidence to support our theoretical findings.  As a proof-of-concept, we conduct a series of experiments to demonstrate the potential applications of these properties in real-world scenarios.

\subsection{Setup}
\label{sec:setup}
Our implementation is built on {\tt Pytorch} and {\tt Torchvision}. We validate our theorems by a range of models and tasks including: the {\tt timm} \footnote{https://github.com/rwightman/pytorch-image-models} version of pre-trained ViT-Base for image classification, the official version of ViT-Adapter \cite{ViTAdapter} for semantic segmentation (pre-trained with DeiT, using UperNet), huggingface's pre-trained Bert and GPT2 for text classification, and GPT2 for text generation. The image tasks are chosen as 10-label classification on Cifar10 \citep{krizhevsky2009learning} consisting of 60,000 natural images, semantic segmentation on ADE20k \cite{zhou2017scene,zhou2019semantic} containing 25,210 images with pixel-level annotations. For text classification, we use IMDB \citep{IMDB} dataset with 50,000 movie reviews for fine-tuning a sentimental classifier, and a natural language inference dataset SNLI\footnote{https://nlp.stanford.edu/projects/snli/} with over 500k sentences. For text generation, we use the GPT2 trained by huggingface on WebText dataset \cite{GPT2} for zero-shot generation on WikiText2 dataset \cite{merity2016pointer}. 

All Transformer backbones are of base size, characterized by a token dimension of 768, 12 heads, and 12 layers. The evaluation metrics are test accuracy for classification tasks, mIoU, aAcc, and mAcc for semantic segmentation, and perplexity for text generation. More detailed setup can be found in Appendix~\ref{sec:exp_setup}.

\subsection{Properties Validation}
\label{sec:property_validation}
\textbf{Row Permutation Forward Equivariance.}
We validate Thm.~\ref{thm:token_permutation_equivariance} by performing inference on `vit\_base\_patch16\_224' pre-trained timm model, the officially released ViT-Adapter, and a small Bert with 2 layers and input of size $(128, 256)$, since the padding mask in huggingface implemented Bert do not follow the properties. The ViT-Base model and the small Bert are fine-tuned on Cifar10 and SNLI respectively before inference. We feed into the same trained model the permuted (Eq.~\ref{eq:training_scenario_permuted}) and normal (Eq.~\ref{eq:training_scenario}) features of test data, respectively.

The results for image tasks are displayed in Tab.~\ref{tab:validation_forwards_CV}. It is evident that the test results for ViT-Base and ViT-Adapter, both with and without row permutation, are identical up to two decimal places. Similarly, the output accuracies of Bert are closely shown in Tab.~\ref{tab:validation_backwards}. 

\begin{table}[t]
    \caption{Test results (\%) of ViT-Base for image classification and ViT-Adapter for segmentation. `Col.' means column and `P.' stands for permutation.}
    \label{tab:validation_forwards_CV}
    \begin{tabular}{lcccc}
    \hline
                      & ViT-Base & \multicolumn{3}{c}{ViT-Adapter} \\
                      & Acc      & aAc      & mIoU    & mAcc   \\ \hline
    Test w/o P.       & 97.65    & 83.12     & 48.75    & 60.09    \\
    Test w/ Row P.  & 97.65    & 83.12     & 48.75    & 60.09    \\
    Test w/ Col. P.   & 97.65    & 83.12     & 48.75    & 60.09    \\ \hline
    \end{tabular}
\end{table}

\begin{table}[]
    \caption{Test accuracy (\%) of ViT-Base for image classification and Bert for text classification. Models are trained by Eq.~\ref{eq:training_scenario} and Eq.~\ref{eq:training_scenario_permuted}.`Col.' means column and `P.' stands for permutation.}
    \label{tab:validation_backwards}
    \centering
    \begin{tabular}{lcc}
    \hline
                      & ViT-Base & Bert (Small) \\ \hline
    Train w/o P.      & 97.75    & 76.4   \\
    Train w/ Row P. & 97.73    & 76.5   \\
    Train w/ Col. P.  & 97.94    & -    \\ \hline
    \end{tabular}
\end{table}

\begin{table}[]
    \caption{Test accuracy (\%) of Bert and GPT2 for text classification, and perplexity of GPT2(G) for text generation. }
    \label{tab:validation_forward_NLP}
    \centering
    \begin{tabular}{lccc}
    \hline
                    & Bert  & GPT2   & GPT2(G) \\ \hline
    Test w/o P.     & 94.00 & 94.03 & 48.55  \\
    Test w/ Col. P. & 94.00 & 94.03 & 48.55  \\ \hline
    \end{tabular}
    \end{table}

\begin{table}[]
    \caption{Test accuracy (\%) of Bert and GPT2 for text classification with models trained by Eq.~\ref{eq:training_scenario} and Eq.~\ref{eq:training_scenario_col_permuted}.  }
    \label{tab:validation_backward_NLP}
    \centering
    \begin{tabular}{lccc}
    \hline
                        & Bert  & GPT2      \\ \hline
    Train w/o P.     & 94.00 & 94.03 \\
    Train w/ Col. P. & 93.72 & 93.66 \\ \hline
    \end{tabular}
\end{table}

\textbf{Row Permutation Backward Equivariance.}
To validate Thm.~\ref{thm:gradient_permutation_backward_Invariance}, we train the pre-trained ViT-Base and the small Bert by Eq.~\ref{eq:training_scenario} and Eq.~\ref{eq:training_scenario_permuted}, respectively. As shown in Tab.~\ref{tab:validation_backwards}, the test results for ViT-Base and Bert with and without row permutation are almost identical, which is consistent with our findings.

\textbf{Col. Permutation Forward Equivariance.}
To validate Thm.~\ref{thm:column_permutation_equivariance}, we test ViT-Base, ViT-Adapter, Bert, and GPT2 on corresponding test datasets. The language models are fine-tuned on the IMDB training set before testing. We first test the model trained by Eq.~\ref{eq:training_scenario} to get the results without permutation, and then permute the trained model by Eq.~\ref{eq:weight_permutation} to get $\mathrm{T}_{(C)}$ in Eq.~\ref{eq:training_scenario_col_permuted}.  The model is reported under the `Test w/ Col. P.'  category in tables.

The results are displayed in Tab.~\ref{tab:validation_forwards_CV} and Tab.~\ref{tab:validation_forward_NLP}. It is evident that the test results for ViT, Bert, and GPT2, both with and without column permutation, are almost identical. Additionally, if we treat the permuted model $\mathrm{T}_{(C)}$ as a normal model and feed normal inputs, the test accuracy falls to that of random guess, i.e., about 10\% on Cifar10 and 50\% on IMDB, indicating power similar to `encryption.' On the text generation task, with or without column permutation, the resulting perplexity remains at $48.55$. If we feed the permuted model $\mathrm{T}_{(C)}$ with normal inputs, the perplexity of the output rises to the order of $10^7$, akin to a randomly initialized language model.

 \textbf{Col. Permutation Backward Equivariance.}
To validate Thm.~\ref{thm:column_permutation_backward_equivariance}, we fine-tune pre-trained ViT-Base, Bert, and GPT2 by Eq.~\ref{eq:training_scenario_col_permuted} with their initial weights permuted by Eq.~\ref{eq:weight_permutation}. Note that if the model is trained from scratch, it is unnecessary to permute its initial weights as they are random. As shown in Tab.~\ref{tab:validation_backwards} and Tab.~\ref{tab:validation_backward_NLP}, the test results for the three models, trained with and without column permutation, are very close with minor gaps due to randomness. Thus all properties are validated.

\subsection{Applications}

In this section, we showcase the potential applications of permutation equivariance property in a real-world context as a proof-of-concept. First, we apply the property to enhancing privacy-preserving split learning in \cite{jeong2022privacy} and \cite{patchshuffling}. Following that, we demonstrate how the column permutation of weights can serve as a model `encryption' or `authorization' tool, which restricts parties without permutation `key' from utilizing the model for inference or fine-tuning.

\begin{figure}[t]
    \centering
    \begin{subfigure}{1\linewidth}
        \centering
      \includegraphics[width=0.7\linewidth]{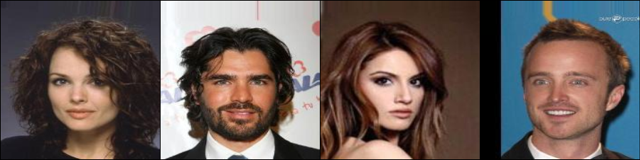}
      \caption{Reconstructed from unprotected features.}
      \label{fig:attack_ori}
    \end{subfigure}
    \\
    \begin{subfigure}{1\linewidth}
        \centering
        \includegraphics[width=0.7\linewidth]{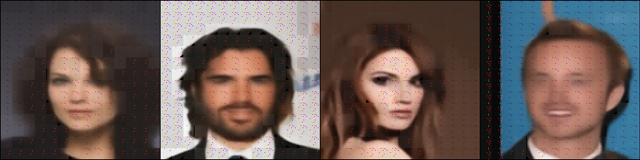}
        \caption{Reconstructed from features protected by GN \cite{patchshuffling}.}
        \label{fig:attack_gn}
      \end{subfigure}
    \\
    \begin{subfigure}{1\linewidth}
        \centering
            \includegraphics[width=0.7\linewidth]{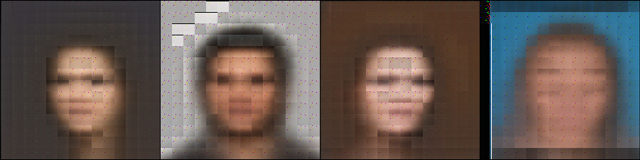}
            \caption{Reconstructed from features protected by GN+ (ours).}
            \label{fig:attack_gn_plus}
          \end{subfigure}
          \\
          \begin{subfigure}{1\linewidth}
            \centering
            \includegraphics[width=0.7\linewidth]{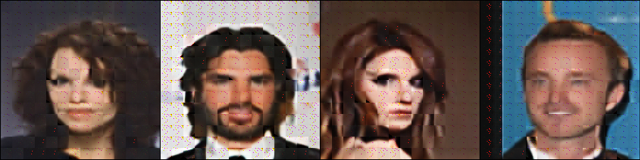}
            \caption{Reconstructed from features protected by LP \cite{jeong2022privacy}.}
            \label{fig:attack_lp}
          \end{subfigure}
          \\
          \begin{subfigure}{1\linewidth}
            \centering
            \includegraphics[width=0.7\linewidth]{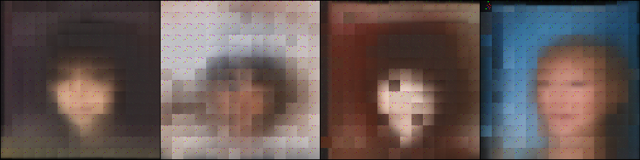}
            \caption{Reconstructed from features protected by LP+ (ours).}
            \label{fig:attack_lp_plus}
          \end{subfigure}
    \caption{Reconstruction results of model inversion attacks to features. `+' means the privacy-preserving technique is enhanced by our row permutation.}
    \label{fig:attack}
\end{figure}

\textbf{Privacy-Preserving Split Learning.}
We implement the split learning framework following \cite{patchshuffling,jeong2022privacy} on a multi-label classification task on CelebA dataset \cite{liu2015deep}, which comprises 2,022,599 images from 10,177 celebrities. In the split learning setting, the server holds the Transformer backbone $T$ while the client holds private data, label, the embedding layer $F_1$, and the task head $F_2$. $F_1$ is computed locally on the client as the input data should be kept private from the untrusted server. $F_2$ is also with the client due to the labels are unknown to the server. The client sends the embeddings to the server and retrieves outputs of $T$ for downstream tasks. The semi-honest server honestly follows the protocol but tries to reconstruct the private data from the embedding by model inversion attack. 

We assess the utility of downstream tasks by classification accuracy and how private the client's embeddings are in protecting the inputs. The privacy is expressed by reconstruction metrics including Structural Similarity (SSIM), Peak Signal to Noise Ratio (PSNR) \citep{hore2010image}, and F-SIM of an MAE \cite{he2022masked} inversion model \cite{dosovitskiy2016inverting, fredrikson2015model, melis2019exploiting}. The worse the reconstruction, the better the privacy protection. For more details of the setup, please see Appendix~\ref{sec:exp_setup}, which aligns with the threat model in \cite{patchshuffling,jeong2022privacy}.

We perform row permutation to enhance the protection level of Gaussian Noise (GN), loss-pass filter (LP) with a radius of 0.05, and Batch Shuffle (BS) techniques. All baselines follow their original implementation. Two row permutation matrices of shape $197 \times 197$ are applied. The results are presented in Tab.~\ref{tab:app_privacy}. As shown, for GN and LP, the row permutation significantly enhances privacy while maintaining the same level of utility as the baselines. For BS, both the utility and privacy are enhanced by our permutation way. The reconstruction visualization effect under model inversion attack is provided in Fig.~\ref{fig:attack} which clearly shows considerable privacy improvement.

To see how permutation affects the trade-off between accuracy and privacy, we report the trade-off curve of LP in Fig.~\ref{fig:app_privacy} by varying the radius $r$ from 0.01 to 0.05. Typically, with the decrease of the radius, the privacy is enhanced with the cost of utility: as the radius is reduced to 0.01, the SSIM, PSNR, and F-SIM scores decrease to 0.190, 7.950, and 0.196, respectively, while the accuracy drops to 85.6\%. With shuffling, one does not need to reduce $r$ to reach a high privacy level. For example, $r = 0.05$ achieves an accuracy of $90.2\%$ under shuffling while reaching privacy no worse than $r = 0.01$ without shuffling.
Overall, with row permutation, LP is able to obtain a better trade-off curve for all ranges of radiuses, as shown in Fig.~\ref{fig:app_privacy}, with high privacy levels at all times.

\begin{table}[]
    \caption{Utility and privacy on CelebA. $\downarrow$ means desirable direction. `+' means an enhanced version by our row permutation. Our methods are marked in light gray.}
    \label{tab:app_privacy}
    \begin{tabular}{lcccc}
    \hline
              & Utility      & \multicolumn{3}{c}{Privacy}        \\ \cline{2-5} 
              & Acc /\% $\uparrow$ & SSIM $\downarrow$ & PSNR $\downarrow$  & F-SIM $\downarrow$\\
    SL                                  & 91.9         & 0.645             & 16.25  & 0.933 \\
    GN \cite{jeong2022privacy}          & 88.8         & 0.457             & 15.47  & 0.663 \\
    \rowcolor{Gray} GN+                                 & 88.8         & 0.167             & 8.068  & 0.187 \\
    LP  \cite{jeong2022privacy}         & 90.1         & 0.450             & 14.71  & 0.614 \\
    \rowcolor{Gray}LP+                                 & 90.2         & 0.110             & 5.687  & 0.154 \\
    BS \cite{patchshuffling}            & 90.8         & 0.148             & 8.180  & 0.176 \\
    \rowcolor{Gray}BS+                                 & 91.1         & 0.143             & 7.613  & 0.167 \\ \hline
    \end{tabular}
    \end{table}

\begin{figure*}
  \centering

  \begin{minipage}{0.32\linewidth}
      \includegraphics[width=1\linewidth]{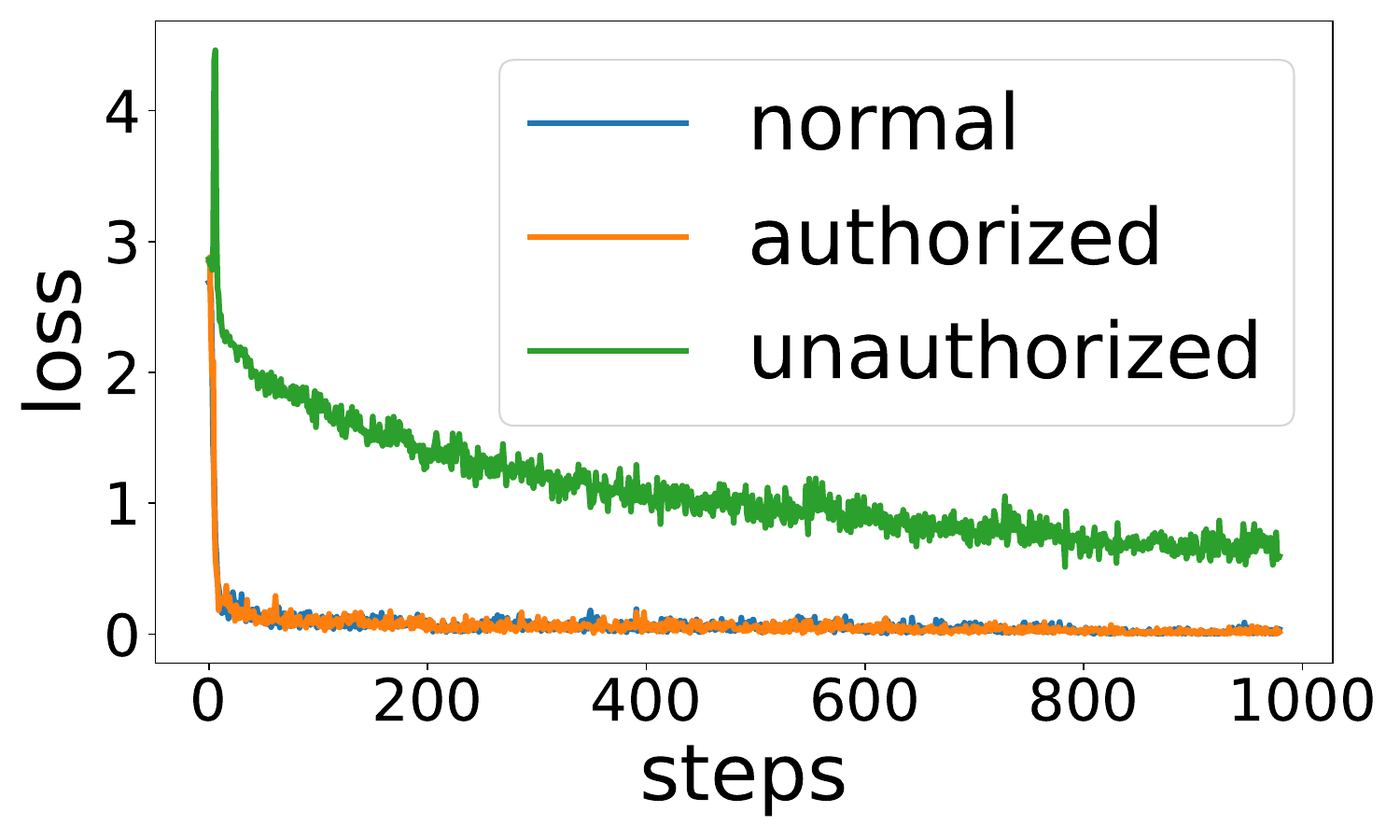}
      \caption{Training curves of fine-tuning ViT. The authorized has a performance close to normal while the unauthorized has a high loss.}
      \label{fig:loss_FT}
  \end{minipage}
  \hfill
  \begin{minipage}{0.32\linewidth}
      \includegraphics[width=1\linewidth]{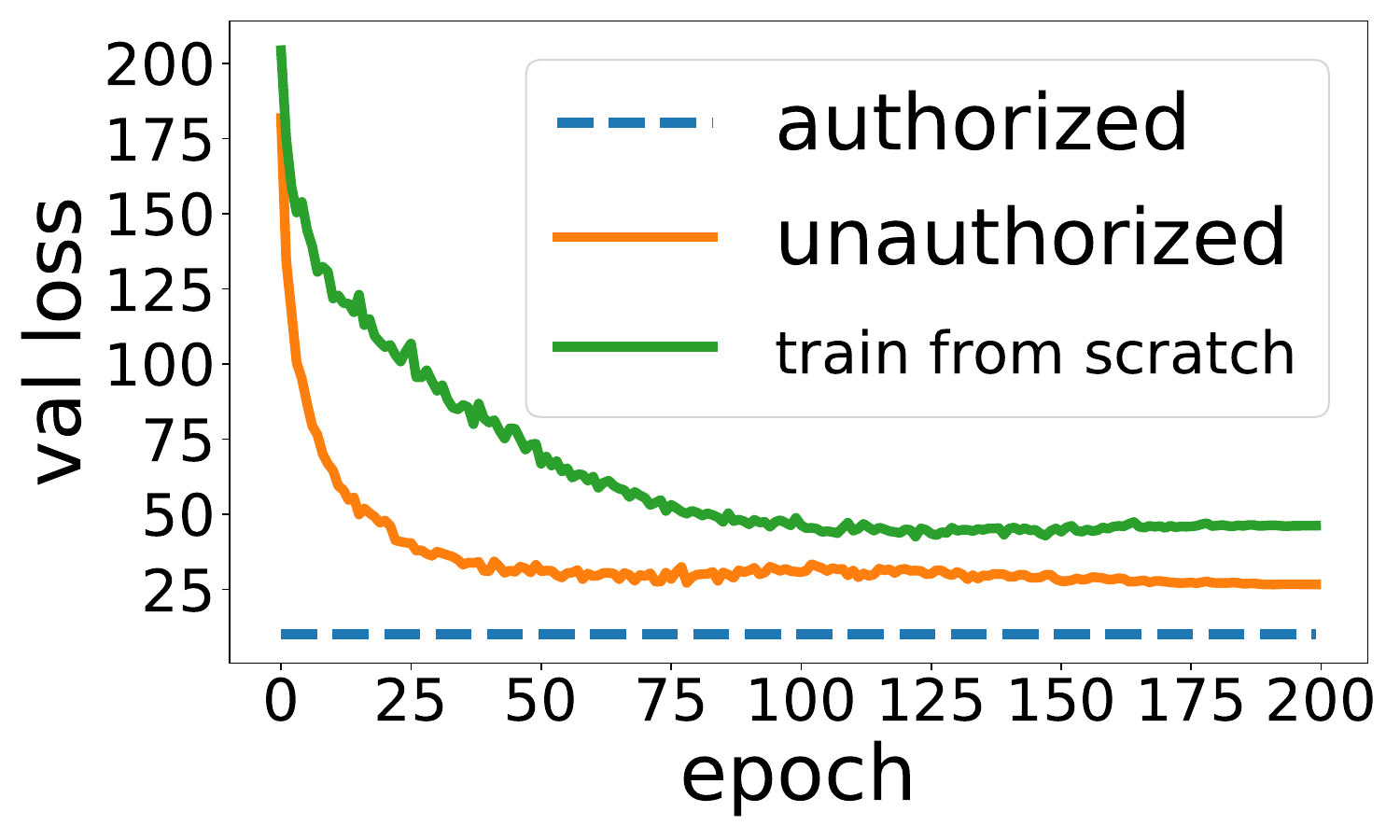}
      \caption{Validation loss curves of ViT trained to convergence. The unauthorized is far worse than the authorized but better than train-from-scratch.  }
      \label{fig:loss_TFS}
  \end{minipage}
  \hfill
  \begin{minipage}{0.32\linewidth}
    \includegraphics[width=1\linewidth]{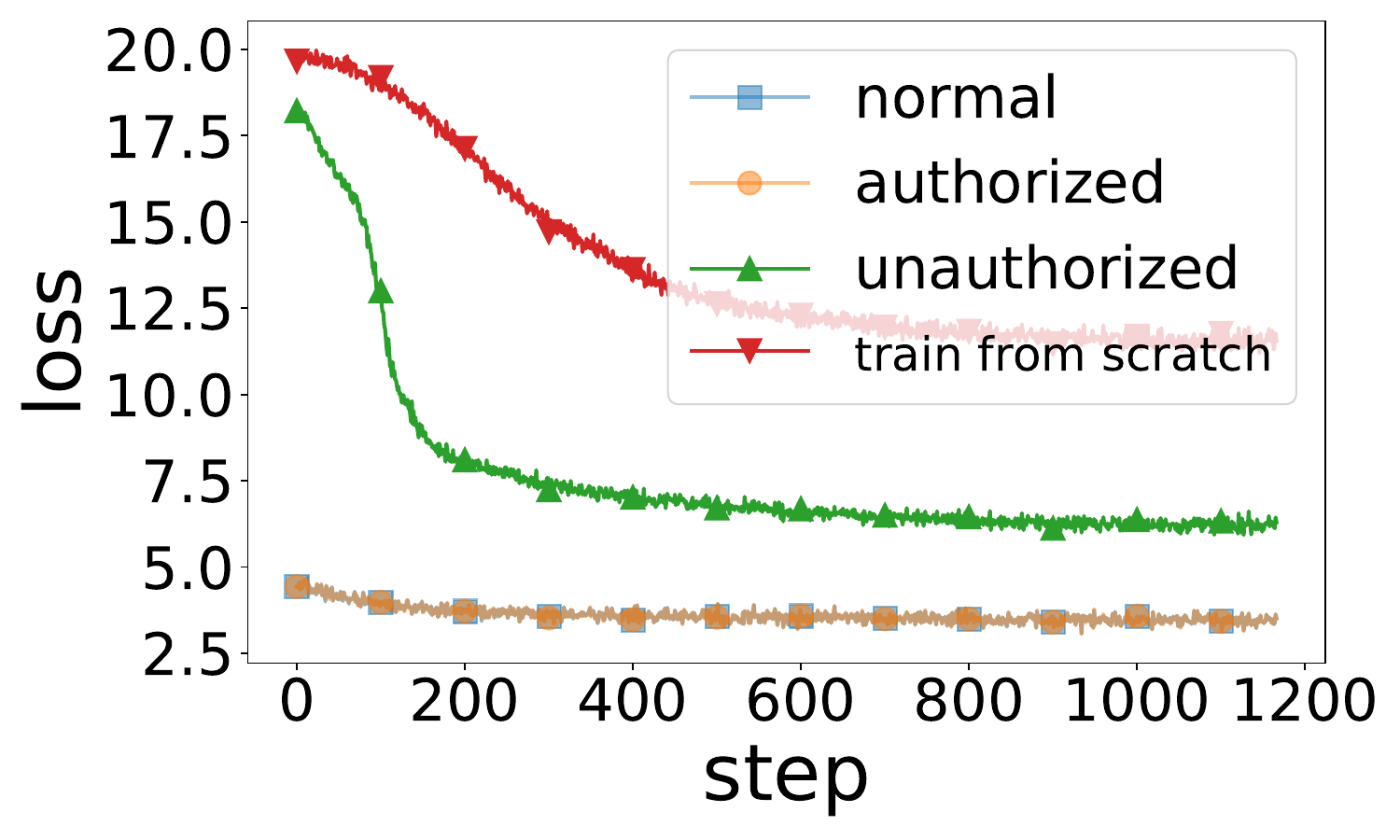}
    \caption{Training curves of fine-tuning GPT2. The unauthorized is far worse than the authorized but better than train-from-scratch. }
    \label{fig:loss_GPT}
\end{minipage}
\end{figure*}
\begin{figure}
  \centering
  \begin{minipage}{0.49\linewidth}
    \includegraphics[width=1.0\linewidth]{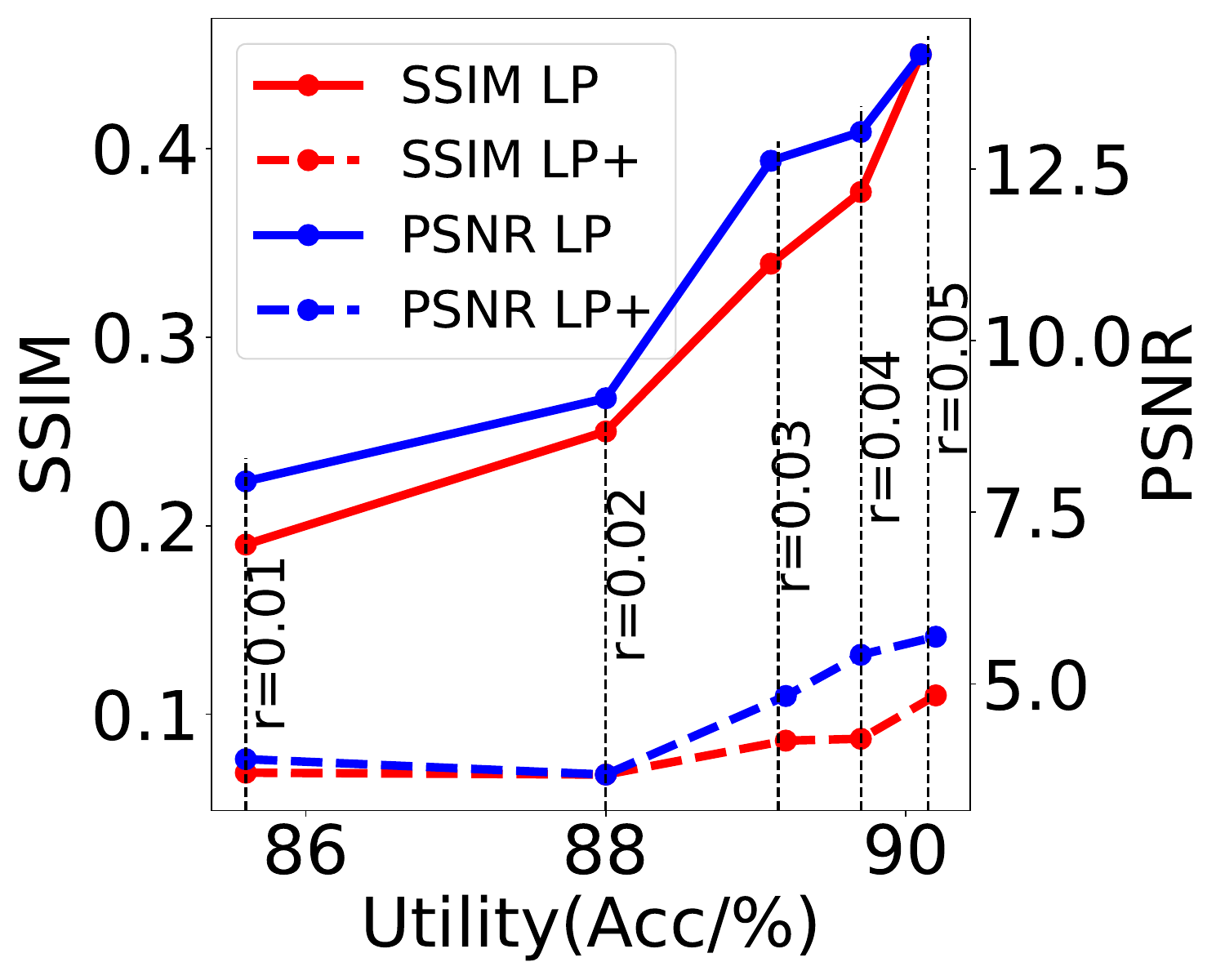}
    \caption{The utility-privacy trade-off curve by varying the radius of LP. The radius is denoted by $r$ in the figure.}
    \label{fig:app_privacy}
  \end{minipage}
  \hfill
  \begin{minipage}{0.49\linewidth}

    \centering
    \includegraphics[width=1.\linewidth]{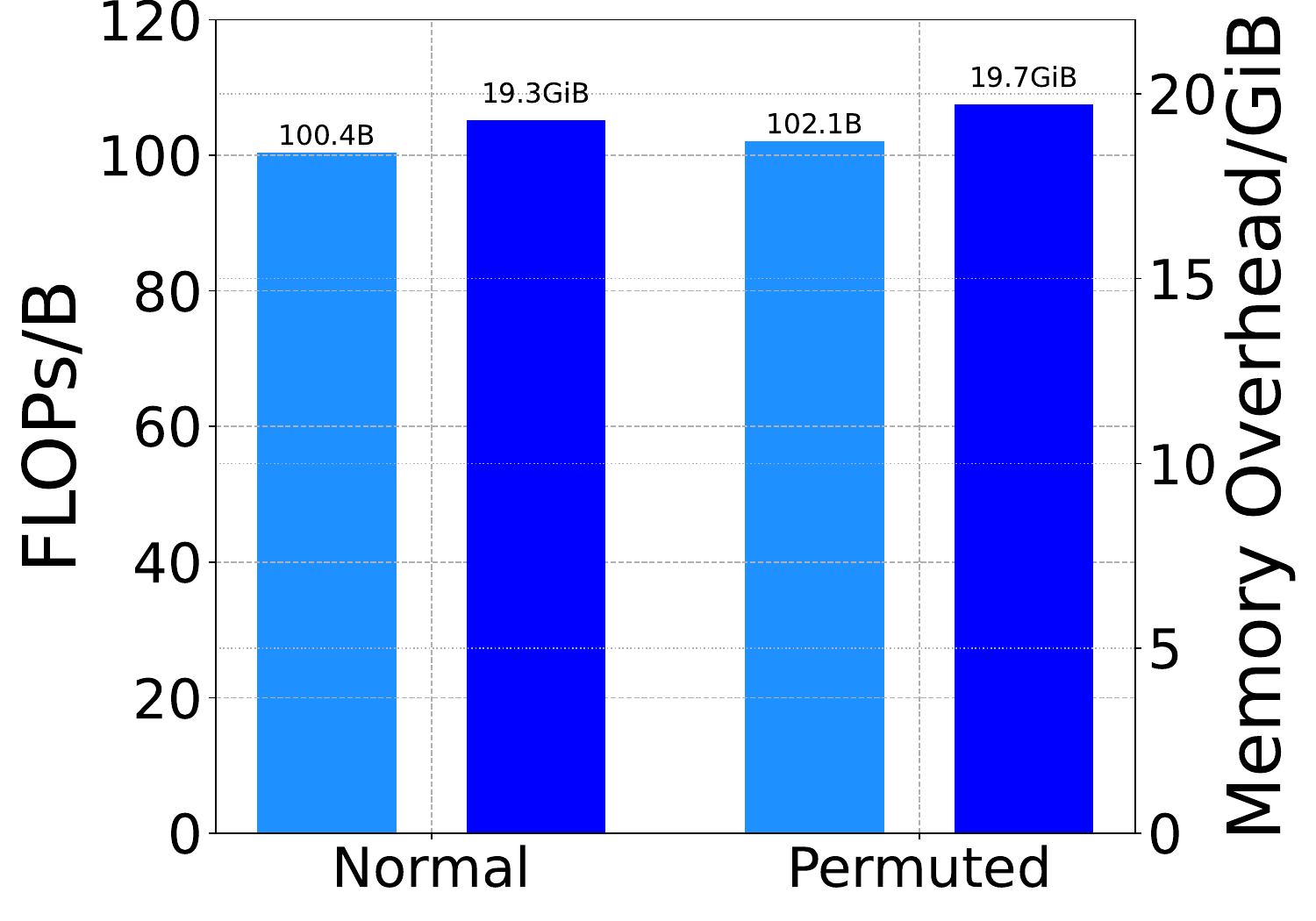}
    \caption{Comparison of computational and memory overhead. }
    \label{fig:cost}
  \end{minipage}
\end{figure}

\textbf{Model Encryption.}
Model parameter leakage, especially for proprietary large language models, could harm the intellectual property of the model owner. We consider an application that `encrypts' the trained model weights by a permutation key, i.e., only the party who knows the permutation matrix can correctly use the model. As previously reported in Sec.~\ref{sec:property_validation}, inference on the `encrypted' model without the column-permutation key is close to random guesses. Further, we show that without the key, fine-tuning the `encrypted' model is also ineffective.

We compare the training curves of three fine-tuning approaches. As a benchmark, the `normal' approach fine-tunes the pre-trained models with normal procedures. Second, the `authorized' approach refers to that the pre-trained model is permuted by Eq.~\ref{eq:weight_permutation} ($\mP_C$ of shape $768 \times 768$) to obtain $\mathrm{T}_{(C)}$ before being fine-tuned with permutation (following Eq.~\ref{eq:training_scenario_col_permuted}). The permutation matrix $\mP_C$ acts as an authorization key in the training process. Third, assume $\mathrm{T}_{(C)}$ is stolen by an unauthorized party who does not own $\mP_C$. The unauthorized party then tries to fine-tune $\mathrm{T}_{(C)}$ with normal procedures.

 The testing accuracy of the fine-tuned models is displayed in Tab.~\ref{tab:application_encryption} with their training curves until 1000 steps are provided in Fig.~\ref{fig:loss_FT} for ViT-Base. The figure shows that the unauthorized can hardly learn as much as the authorized. To see if the unauthorized would further improve with a better training strategy, we train both the authorized and unauthorized until full convergence (of the authorized) in Fig.~\ref{fig:loss_TFS}. The train-from-scratch model is trained with the same strategy, suggesting the lower bound of the unauthorized: in the worst case, the unauthorized party has to train random weights from the start, indicating no use of the `encrypted' model.
 
It is observed that the performance of the unauthorized sits in between the authorized and that of training from scratch. A similar trend is observed in fine-tuning GPT2 on WikiText2 for text generation. 
The convergence curves are given in Fig.~\ref{fig:loss_GPT} where the authorized behaves almost identically to the normal. The unauthorized has a curve in between the normal and the train-from-scratch. Hence one can conclude that permutation indeed prevents anyone who is unaware of the permutation matrices from taking advantage of a trained model for inference or from releasing the full power of the pre-trained model by fine-tuning.

\begin{table}[]
    \caption{Test accuracy (\%) of the normal models fine-tuned in normal procedure, `encrypted' models fine-tuned in authorized and unauthorized manners. }
    \label{tab:application_encryption}
    \centering
    \begin{tabular}{lccc}
    \hline
     Model-FineTuning  & ViT   & Bert  & GPT2  \\ \hline
    Normal-Normal   & 97.74 & 94.00 & 94.03 \\
    Encrypted-Authorized & 97.94 & 93.72 & 93.66 \\
    Encrypted-Unauthorized  & 78.03 & 77.25 & 86.16  \\ \hline
    \end{tabular}
    \end{table}

  

  \subsection{Efficiency}   
  To see how much additional overhead our method incurs to vanilla learning, we evaluate the efficiency of our method by floating point operations (FLOPs) of the Transformer-Base backbone, and the memory consumption of training ViT-Base on Cifar10 in the normal and the permuted settings, respectively. For a more precise comparison, exactly four more additional matrix multiplications are applied to the permuted setting: two row permutation by $\mP_R\in \mathbb{R}^{197 \times 197}$ and two column permutation by $\mP_C\in \mathbb{R}^{768 \times 768}$. $\mP_R$ is sampled per batch whereas $\mP_C$ is fixed for the model. Results of Fig.~\ref{fig:cost} show that our method are almost as efficient as normal learning, incurring negligible overhead in computation and memory consumption.

\section{Conclusion}
\label{sec:conclusion}
We revise the concept of permutation equivariance and prove it as a property for Transformer-based models. The permutation equivariance of prior works is only a subset of ours: we show inter- and intra- token permutation equivariance in both forward and backward propagation. We not only theoretically analyze the property but also validate it through experiments. Finally, we propose two applications of this property: privacy-enhancing split learning and model authorization. Our research contributes to a deeper understanding of Transformer-based models and their wide potential applications.


{
    \small
    \bibliographystyle{ieeenat_fullname}
    \bibliography{main}
}

\clearpage
\setcounter{page}{1}
\maketitlesupplementary


\section{Structure of Transformer}\label{sec:encoder}

\textbf{Transformer-based models} are the state-of-the-art deep neural networks and have attracted great attention in both areas of computer vision and natural language processing. Models including transformer encoder blocks as their backbone, such as Bert \cite{Bert}, ViT \cite{ViT}, T2T-ViT \cite{yuan2021tokens}, ViTGAN \cite{hirose2021vit}, BEiT \cite{BEiT} and CoCa \cite{yu22coca}, have been achieving exceeding performance in a great many tasks.

Transformer encoder blocks, as shown in Fig.~\ref{tranenc}, mainly contain two critical components: Multi-head Scaled-dot-product self-attention and a feed-forward network (MLP). Inputs are fed in the form of patches, which are usually embedding vectors for words in Bert, or for fractions of images in ViT. The relative position of patches are learned by position embeddings \cite{Transformer}, which are injected into the model. Fig.~\ref{tranenc} shows the main operators in a Transformer where the shortcut and the linear projection in the Attention block are left out for simplicity. 

\begin{figure}[h]
	\centering
	\includegraphics[scale = 0.6]{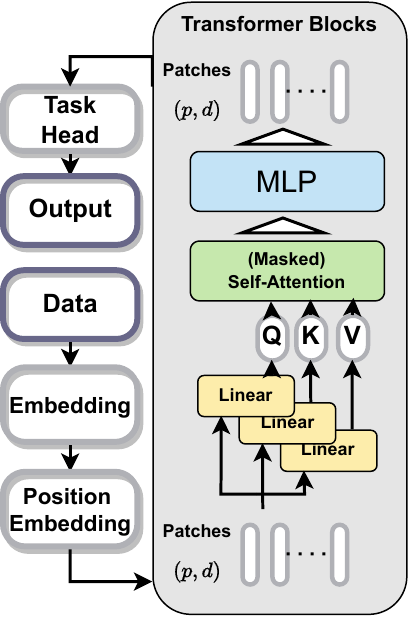}
	\caption{Transformer Encoder Block}
	\label{tranenc}
\end{figure}

The Transformer encoder block is denoted as $\mathrm{Enc}$ and the loss is $\ell$. The patch embedding of a single input $\mX$ is expressed as $\mZ$ of shape $(p,d)$. The first layer in the self-attention contains three parallel linear layers projecting $\mZ$ to $Q, K, V$ as
\begin{align}
	\mQ = \mZ \mW_Q^{\top},\\
	\mK = \mZ \mW_K^{\top},\\
	\mV = \mZ \mW_V^{\top}.
\end{align}
$Q, K, V$ are fed to the following attention operation
\begin{align}
	\mS&={Softmax}(\frac{\mQ\mK^{\top}}{\sqrt{d}}),\\
	\mA&=\mS\mV,
\end{align}
where $\mS$ and $\mA$ are the softmax output, and the attention output, respectively. 

We neglect the attention projection and the residual connection for simplicity. 
The part following the attention layer is the MLP layer:
\begin{align}
	\mA_1 &= \mA\mW_1^{\top},\\
	\mH &= a(\mA_1),\\
	\mA_2 &= \mH \mW_2^{\top}
\end{align}
where $\mA_1,\mA_2$ are the outputs of the linear layers with weights $\mW_1,\mW_2$, respectively, and $\mH$ is the output of the element-wise activation function $a$ which can be ReLu, Tanh, etc.

The backward propagation of Transformer encoder block is as following , we calculate the all the gradients from the final layer back to the first. Gradients are expressed as
\begin{align*}
	\mathrm{d}l  &=\mathrm{tr}(\frac{\partial l}{\partial \mA_2}^{\top}\mathrm{d}\mA_2)\\
	&=\mathrm{tr}(\frac{\partial l}{\partial \mA_2}^{\top}(\mathrm{d}\mH)\mW_2^{\top})+\mathrm{tr}(\frac{\partial l}{\partial \mA_2}^{\top}\mH \mathrm{d}(\mW_2^{\top})).
\end{align*}

The two additive terms are inspected in the following. Let's study $\mH$ first:
\begin{align*}
	\mathrm{d}l_1&\triangleq \mathrm{tr}(\frac{\partial l}{\partial \mA_2}^{\top}(\mathrm{d}\mH)\mW_2^{\top})\\
	&=\mathrm{tr}(\mW_2^{\top}\frac{\partial l}{\partial \mA_2}^{\top}\mathrm{d}\mH)\\
	&=\mathrm{tr}((\frac{\partial l}{\partial \mA_2}\mW_2)^{\top}\mathrm{d}\mH),
\end{align*}
indicating
\begin{equation}\label{h}
	\frac{\partial l}{\partial \mH} = \frac{\partial l}{\partial \mA_2}\mW_2.
\end{equation}

For $\mW_2$,
\begin{align*}
	\mathrm{d}l_2&\triangleq \mathrm{tr}(\frac{\partial l}{\partial \mA_2}^{\top}\mH \mathrm{d}(\mW_2^{\top}))\\
	&=\mathrm{tr}(\mathrm{d}\mW_2\mH^{\top}\frac{\partial l}{\partial \mA_2})\\
	&=\mathrm{tr}((\frac{\partial l}{\partial \mA_2}^{\top}\mH)^{\top}\mathrm{d}\mW_2),
\end{align*}
and 
\begin{equation}\label{w2}
	\frac{\partial l}{\partial \mW_2}=\frac{\partial l}{\partial \mA_2}^{\top}\mH.
\end{equation}

For $\mA_1$:
\begin{align*}
	\mathrm{d}l_1&=\mathrm{tr}((\frac{\partial l}{\partial \mH}^{\top}\mathrm{d}\mH)\\
	&=\mathrm{tr}(\frac{\partial l}{\partial \mH}^{\top} \mathrm{d}(a(\mA_1)))\\
	&=\mathrm{tr}(\frac{\partial l}{\partial \mH}^{\top} a'(\mA_1)\odot \mathrm{d}\mA_1))\\
	&=\mathrm{tr}((\frac{\partial l}{\partial\mH}\odot a'(\mA_1))^{\top} \mathrm{d}\mA_1),
\end{align*}
by Eq.~\ref{h}, we have 
\begin{equation}\label{a1}
	\frac{\partial l}{\partial \mA_1}=\frac{\partial l}{\partial \mA_2}\mW_2 \odot a'(\mA_1).
\end{equation}

Similarly, we calculate the gradients of $\mA$ and $\mW_1$:
\begin{equation}\label{A}
	\frac{\partial l}{\partial \mA}=\frac{\partial l}{\partial \mA_1}\mW_1,
\end{equation}
\begin{equation}\label{W1}
	\frac{\partial l}{\partial \mW_1}=\frac{\partial l}{\partial \mA_1}^{\top}\mA.
\end{equation}

In the attention operation:
\begin{align*}
	\mathrm{d}l_3&\triangleq \mathrm{tr}(\frac{\partial l}{\partial \mA}^{\top}\mathrm{d}\mA)\\
	&=\mathrm{tr}(\frac{\partial l}{\partial \mA}^{\top} (\mathrm{d}\mS)\mV)+\mathrm{tr}(\frac{\partial l}{\partial A}^{\top} \mS\mathrm{d}\mV)\\
	&=\mathrm{tr}((\frac{\partial l}{\partial \mA}\mV^{\top})^{\top} \mathrm{d} \mS)+\mathrm{tr}((\mS^{\top}\frac{\partial l}{\partial \mA})^{\top}\mathrm{d} \mV),
\end{align*}
and
\begin{equation}\label{S}
	\frac{\partial l}{\partial \mS}=\frac{\partial l}{\partial \mA}\mV^{\top},
\end{equation}
\begin{equation}\label{V}
	\frac{\partial l}{\partial \mV}=\mS^{\top}\frac{\partial l}{\partial \mA}.
\end{equation}

First, for $\mV =\mZ\mW_V^{\top}$:
\begin{align*}
	\mathrm{d}l_4&\triangleq \mathrm{tr}(\frac{\partial l}{\partial \mV}^{\top}\mathrm{d}\mV)\\
	&=\mathrm{tr}(\frac{\partial l}{\partial \mV}^{\top}(\mathrm{d}\mZ)\mW_V^{\top})+\mathrm{tr}(\frac{\partial l}{\partial \mV}^{\top}\mZ\mathrm{d}\mW_V^{\top}).\\
\end{align*}
Similarly, the gradients of $\mZ$ and $\mW_V$ are:
\begin{equation}\label{X}
	\frac{\partial l}{\partial \mZ}=\frac{\partial l}{\partial \mV}\mW_V,
\end{equation}
\begin{equation}\label{WV}
	\frac{\partial l}{\partial \mW_V}=\frac{\partial l}{\partial \mV}^{\top}\mZ.
\end{equation}

Now we focus on $\mS =Softmax(\frac{\mQ\mK^{\top}}{\sqrt{d} } ) $:
\begin{align*}
	\mathrm{d}l_5& \triangleq \mathrm{tr}(\frac{\partial l}{\partial \mS}^{\top}\mathrm{d}\mS)\\
	&=\mathrm{tr}(\frac{\partial l}{\partial \mS}^{\top}(diag(\mS)-\mS^{\top} \mS)\mathrm{d}(\frac{\mQ\mK^{\top}}{\sqrt{d}}))\\
	&=\mathrm{tr}(((diag(\mS)-\mS^{\top} \mS)^{\top}\frac{\partial l}{\partial \mS})^{\top}\mathrm{d}(\frac{\mQ\mK^{\top}}{\sqrt{d} } )), \\
\end{align*}
and thus
\begin{equation}\label{Q}
	\frac{\partial l}{\partial \mQ}=\frac{1}{\sqrt{d}}((diag(\mS)-\mS^{\top} \mS)^{\top}\frac{\partial l}{\partial \mS}) \mK,
\end{equation}
\begin{equation}\label{K}
	\frac{\partial l}{\partial \mK}=\frac{1}{\sqrt{d}}((diag(\mS)-\mS^{\top} \mS)^{\top}\frac{\partial l}{\partial \mS})^{\top} \mQ.
\end{equation}

And similarly the gradients of $\mW_Q$ and $\mW_K$ are:
\begin{equation}\label{WQ}
	\frac{\partial l}{\partial \mW_Q}=\frac{\partial l}{\partial \mQ}^{\top}\mZ,
\end{equation}
\begin{equation}\label{WK}
	\frac{\partial l}{\partial \mW_K}=\frac{\partial l}{\partial \mK}^{\top}\mZ.
\end{equation}

\section{Alg. on Permuted Training}\label{sec:elaboration_shuffle}
{Our permuted training} is described by pseudo code in Alg.~\ref{alg1}. It should be noted that the permutation takes place not on the dimension of `batches' but on the rest two dimensions. Taking ViT for example, each image is transformed into a $(p, d)$ matrix representing $p$ patches, and each patch denotes a fraction of the image. Each fraction is embedded into a $d$-dimensional vector.
\begin{algorithm}
	\renewcommand{\algorithmicrequire}{\textbf{Input:}}
	\renewcommand{\algorithmicensure}{\textbf{Output:}}
	\caption{Permuted Training}
	\label{alg1}
	\begin{algorithmic}[1]
		\STATE Initialization: Initialize the model. Load permutation matrices $\mP_R, \mP_C$.
		\STATE Start training
		\REPEAT
		\STATE Start a new epoch
		\REPEAT 
		\STATE Get a batch of data $\mX$ from data loader

		\STATE Get embedding $\mZ$ of size $(batch\_size,p,d)$.
        
		\IF{using row permutation}
		\STATE $\mZ = $~matmul($\mP_R, \mZ$)
		\ENDIF
		\IF{using column permutation}
		\STATE $\mZ = $~matmul($\mZ, \mP_C$)
		\ENDIF
        
		\STATE  Send $\mZ$ to the Transformer Backbone and retrieve the output $\hat{\mY}$

        \IF{using row permutation}	
		\STATE $\hat{\mY} = $~matmul($\mP_R^{-1}, \hat{\mY}$)
		\ENDIF
		\IF{using column permutaton}
		\STATE $\hat{\mY} = $~matmul($\hat{\mY}, \mP_C^{-1}$)
		\ENDIF
        
		\STATE Perform backward propagation
		\UNTIL done all batches
		\UNTIL done all epochs
	\end{algorithmic}  
\end{algorithm}

We further provide a toy example. Let $\mZ$ of shape $(3,4)$ and the row shuffle matrix $\mP_R$ be

\begin{equation*}
   \mZ = 
   \begin{pmatrix}
  1& 2 & 3 & 4\\
  5& 6 & 7 & 8\\
  9& 10 & 11 &12
\end{pmatrix}
~~~~
   \mP_R = 
   \begin{pmatrix}
  0& 1 & 0 \\
  0& 0 & 1 \\
  1& 0 & 0 
\end{pmatrix}.
\end{equation*}
The row permuted feature is
\begin{equation*}
    \mP_R\mZ = 
    \begin{pmatrix}
  5& 6 & 7 & 8\\
  9& 10 & 11 &12\\
  1& 2 & 3 & 4\\
    \end{pmatrix}.
\end{equation*}

Column permutation is performed in a similar way. Note that permutation matrices are orthogonal, i.e., $\mP_R^{-1} = \mP_R^{\top}$ and $\mP_C^{-1} = \mP_C^{\top}$.

\section{Permutation-Equivariant Operators}
\label{sec:permutation_equivariant_operators}

As far as we will show, the following operators are permutation-equivariant:
\begin{itemize}
    \item Element-wise operators,
    \item Softmax,
    \item Linear layer, 
    \item MLP,
    \item LayerNorm and BatchNorm,
    \item Attention.
\end{itemize}
In the following, we will prove the permutation-equivariance of each operator.

\textbf{Element-wise operators} including shortcut, Hadamard product, matrix addition/subtraction and other element-wise functions. We have
\begin{lemma}\label{lem:element_wize}
Element-wise operators are permutation-equivariant that
\begin{equation}\label{Eq_element_wise}
    (\mP_R\mA\mP_C) \odot (\mP_R\mB\mP_C) = \mP_R(\mA\odot\mB)\mP_C.
\end{equation}
where $\odot$  denotes the element-wise operation. 
\end{lemma}
On the left hand-side of the equation, $a_{ij}$ in $\mA$ and $b_{ij}$ in $\mB$ are permuted to the same position before being performed the operation. On the right hand-side, $a_{ij}$ and $b_{ij}$ are performed the operation of which the results are permuted. The two are obviously equivariant. Lemma~\ref{lem:element_wize} also holds for matrix addition and activation function:
\begin{equation}
    a(\mP_R\mA\mP_C) = \mP_R a(\mA) \mP_C
\end{equation}
where $a$ is an element-wise activation function, or other element-wise functions like scalar multiplication, division, etc.

\begin{lemma}\label{softmax}
Softmax is permutation-equivariant:
\begin{equation}
    Softmax(\mP_R\mA\mP_C) = \mP_R Softmax(\mA) \mP_C.
\end{equation}
\end{lemma}
This is because an element is always normalized with the same group of elements, which are not changed in permutations. Thus Softmax is permutation-equivariant.

\begin{lemma}\label{Lemma_linear}
Linear layer is permutation-equivariant:
\begin{equation}
    f_{(P)}(\mP_R \mX \mP_C) = \mP_R f(\mX) \mP_C
\end{equation}
where $f(\mX) = \mX\mW^{\top}+b$ and $f_{(P)}(\mX) =\mX\mW_{(P)}^{\top}+b_{(P)}$ and:
\begin{equation*}
    \mW_{(P)} = \mP_C^{\top}\mW\mP_C,
\end{equation*}
\begin{equation*}
    b_{(P)} = b\mP_C.
\end{equation*}
\end{lemma}
\begin{proof}
\begin{align*}
    f_{(P)}(\mP_R \mX \mP_C) 
    &= \mP_R \mX \mP_C\mW_{(P)}^{\top} + b_{(P)}\\
    &= \mP_R \mX \mP_C\cdot \mP_C^{\top}\mW^{\top}\mP_C+b\mP_C\\
    &=\mP_R \mX\mW^{\top}\mP_C+b\mP_C\\
    &= \mP_R f(\mX) \mP_C,
\end{align*}
where the bias $b$ is broadcast to each row. Note that if $\mP_C \neq \mI$, the identity matrix, this lemma is limited to linear layer with square weight matrix. If $\mP_C$ is not included, i.e. only row shuffle is used, then all linear layers are row-permutation equivariant.
\end{proof}

\begin{lemma}\label{lem:MLP}
    MLP is permutation-equivariant:
    \begin{equation}
        f_{(P)}(\mP_R \mX \mP_C) = \mP_R f(\mX) \mP_C
    \end{equation}
    where 
    \begin{equation*}
        f(\mX) = \sigma(\mX\mW_1^{\top}+b_1)\mW_2^\top + b_2,
    \end{equation*} 
    \begin{equation*}
        f_{(P)}(\mX) =\sigma(\mX\mW_{1(P)}^{\top}+b_{1(P)})W_{2(P)}^\top + b_{2(P)},
    \end{equation*}
    and:
    \begin{align*}
            &\mW_{1(P)} = \mW_1\mP_C,
            \mW_{2(P)} = \mP_C^{\top}\mW_2,\\
            &b_{1(P)} = b_1, 
            b_{2(P)} = b_2\mP_C,
    \end{align*}
    where $\sigma$ is the activation function, $\mW_1 \in \mathbb{R}^{t\times d}$, $\mW_2 \in \mathbb{R}^{d\times t}$, $b_1 \in \mathbb{R}^{t}$, $b_2 \in \mathbb{R}^{d}$, and $t$ is the hidden dimension of MLP.
\end{lemma}
\begin{proof}
    \begin{align*}
        f_{(P)}(\mP_R \mX \mP_C) &= \sigma(\mP_R \mX \mP_C \mW_{1(P)}^\top + b_{1(P)}) \mW_{2(P)}^\top + b_{2(P)}\\
        &= \sigma(\mP_R \mX \mW_1^\top  + b_1) \mW_2^\top\mP_C + b_2\mP_C\\
        &= \mP_R( \sigma(\mX \mW_1^\top  + b_1) \mW_2^\top+ b_2)\mP_C\\
        &= \mP_R f(\mX) \mP_C
    \end{align*}
where the third equation holds due to Lem.~\ref{lem:element_wize} and the broadcast of bias.
\end{proof}

\begin{lemma}
    Normalization (LayerNorm for example, LN for short) is permutation-equivariant:
    \begin{equation}
        \mathrm{LN}_{(P)}(\mP_R \mX \mP_C)=\mathrm{LN}(\mX)
    \end{equation}
    where $\mathrm{LN}(\mX)=\frac{\mX-\mathrm{E}(\mX)}{\sqrt{\mathrm{Var}(\mX)-\epsilon}}\ast \gamma + b$, and:
    \begin{align*}
        \gamma_{(P)} = \gamma\mP_C,~~~
        b = b\mP_C.
    \end{align*}
\end{lemma}
Since the same $\mathrm{E}(\mX)$ and $\mathrm{Var}(\mX)$ work on each element, permutation dose not affect the normalization operation. And the affine operation is `column-wise', weight $\gamma$ and bias $b$ are broadcast to each row.

\begin{lemma}\label{lem_attetion}
    Attention ($\mA = Softmax(\frac{\mQ\mK^{\top}}{\sqrt{d}})\mV$) is permutation-equivariant:
    \begin{equation}
        \begin{split}
            \mathrm{Attetion}(\mP_R\mQ\mP_C, \mP_R\mK\mP_C, \mP_R\mV\mP_C) \\ = \mP_R \mathrm{Attetion}(Q, K, V) \mP_C.
        \end{split}
    \end{equation}
\end{lemma}
\begin{proof}
    \begin{align*}
        &~~\mathrm{Attetion}(\mP_R\mQ\mP_C, \mP_R\mK\mP_C, \mP_R\mV\mP_C) \\
        &= Softmax(\frac{\mP_R\mQ\mP_C\cdot\mP_C^{\top}\mK^{\top}\mP_R^{\top}}{\sqrt{d}})\mP_R\mV\mP_C\\
        &= \mP_R Softmax(\frac{\mQ\mK^{\top}}{\sqrt{d}}) \mP_R^{\top}\cdot \mP_R\mV\mP_C\\
        &= \mP_R \mathrm{Attetion}(Q, K, V) \mP_C
    \end{align*}
    where the second equality holds because of the permutation equivariance of Softmax.
\end{proof}

Multihead attention is a special case. The validity of Thm.~\ref{thm:column_permutation_equivariance} and Thm.~\ref{thm:column_permutation_backward_equivariance} is contingent on constraining the permutation $\mP_C$ to operate within a single head. It means permutation equivariance holds if permutation is performed within each head, or on different heads, but not across heads. In that case, the feasible permutation space shrinks but the space is still considerable. Take a base Transformer for example, the possible permutations is reduced from $768!$ to $12!\times 64!$.

We later provide detailed proofs of Thm.~\ref{thm:token_permutation_equivariance}, Thm.~\ref{thm:gradient_permutation_backward_Invariance}, Thm.~\ref{thm:column_permutation_equivariance}, Thm.~\ref{thm:column_permutation_backward_equivariance} specifically for Transformer architecture.

\section{Proofs on Transformer Encoder Blocks}
\label{sec:proofEnc}

We show the detailed proof on the Transformer encoder blocks. The notations are shown in Appendix~\ref{sec:encoder}, and Transformer Encoder Block is denoted as $\mathrm{Enc}$ for short. 

\subsection{Enc is Forward Permutation-Equivariant}\label{proof_RC_forward}

$\mathrm{Enc}$ is forward permutation-equivariant. As proven above, all the basic operators in $\mathrm{Enc}$ are permutation-equivariant. The following section shows how the combination of the operators still holds in detail. The proofs of Thm.~\ref{thm:token_permutation_equivariance} and Thm.~\ref{thm:column_permutation_equivariance} are organized into one where the weight matrices are permuted by $\mP_R$ and $\mP_C$ at the same time. The row permutation equivariance can be seen as a special case where $\mP_C = \mI$, and the column permutation equivariance is a special case of $\mP_R = \mI$.

\begin{proof}
First and foremost, we `encrypt' all the weight matrices by Eq.~\ref{eq:weight_permutation}:
\begin{equation*}
    \mW_{i(p)}=\mP_C^{\top}\mW_i\mP_C,
\end{equation*}
where $\mP_C$ is the column permutation matrix, $\mW_i$ is the weight of a normal $\mathrm{Enc}$, and $i \in \{Q,K,V\}$. Weights in MLP are `encrypted' by $\mW_{1(C)} = \mW_1\mP_C,~\mW_{2(C)} = \mP_C^{\top}\mW_2$. We denote the Transformer encoder block with such `encryption' as $\mathrm{Enc}_{(P)}$. 

For $\mQ$:
\begin{align}
    \mQ_{(P)}&=\mZ_{(P)}\mW_{\mQ(P)}^{\top}\\
    &=\mP_R\mZ\mP_C\cdot \mP_C^{\top}\mW_{\mQ}^{\top}\mP_C\\
    &=\mP_R\mZ\mW_Q^{\top}\mP_C\\
    &=\mP_R\mQ\mP_C.
\end{align}
Similarly for $\mK,\mV$:
\begin{align}
    \mK_{(P)}&=\mP_R\mK\mP_C,\\
    \mV_{(P)}&=\mP_R\mV\mP_C.
\end{align}

For $\mS=Softmax(\frac{\mQ\mK^{\top}}{\sqrt{d}})$:
\begin{align}
    \mS_{(P)}&=Softmax(\frac{\mQ_{(P)}\mK_{(P)}^{\top}}{\sqrt{d}})\\
    &=Softmax(\frac{\mP_R\mQ\mP_C\cdot \mP_C^{\top}\mK^{\top}\mP_R^{\top}}{\sqrt{d}})\\
    &=Softmax(\frac{\mP_R\mQ\mK^{\top}\mP_R^{\top}}{\sqrt{d}})\\
    &=\mP_RSoftmax(\frac{\mQ\mK^{\top}}{\sqrt{d}})\mP_R^{\top}\\
    &=\mP_R\mS\mP_R^{\top}. \label{sforward}
\end{align}

So for $\mA$:
\begin{align}
    \mA_{(P)}&=\mS_{(P)}\mV_{(P)}\\
    &=\mP_R\mS\mP_R^{\top}\cdot \mP_R\mV\mP_C\\
    &=\mP_R\mS\mV\mP_C\\
    &=\mP_R\mA\mP_C.
\end{align}

Following the attention layer, $\mA$ is fed to the MLP layer:
\begin{align}
    \mA_{1(P)}&=\mA_{(P)}\mW_{1(P)}^{\top}\\
    &=\mP_R\mA\mP_C \cdot \mP_C^{\top}\mW_1\\
    &=\mP_R\mA\mW_1\\
    &=\mP_R\mA_1.
\end{align}
Similarly for $\mA_2$,
\begin{equation}
\mA_{2(P)}=\mP_R\mA_2\mP_C.
\end{equation}

As for the activation in the middle, the element-wise activation function is permutation-equivariant:
\begin{equation}
    \mH_{(P)}=\mP_R\mH.
\end{equation}

Overall, we have proved $\mathrm{Enc}$ satisfies permutation forward  equivariance.
\end{proof}

\subsection{Enc is Backward Permutation-Invariant}\label{proof_RC_backward}
According to  Thm.~\ref{thm:general_permutation_equivariance}, since all the operators in $\mathrm{Enc}$ are forward permutation-equivariant, the feature of $\mathrm{Enc}$ is backward permutation-equivariant and the weight in $\mathrm{Enc}$ is permutation-invariant. The following section shows how the combination of the operators still holds in detail. Similar to the proof of forward permutation equivariance, we prove Thm.~\ref{thm:gradient_permutation_backward_Invariance} and Thm.~\ref{thm:column_permutation_backward_equivariance} altogether in one proof where the weight matrices are permuted by $\mP_R$ and $\mP_C$ at the same time. The row permutation equivariance can be seen as a special case where $\mP_C = \mI$, and the column permutation equivariance is a special case of $\mP_R = \mI$.

\begin{proof}
Due to the shuffling and unshuffling procedures of Alg.~\ref{alg1}, we have the forward and backward propagation outside of the backbone no different from the normal ones. Hence we only focus on the propagation of the Transformer encoder blocks.

We denote $\mA_{3(P)}$ as the reversed intermediate feature that the down-stream head receives: 
\begin{equation}
    \mA_{3(P)} = \mP_R^{\top} \mA_{2(P)} \mP_C^{\top}.
\end{equation}
Since the feature is unshuffled, we have 
\begin{equation}\label{a2a3}
    \mA_{3(P)} =\mA_3 = \mA_2.
\end{equation}

First, we focus on the MLP layer:
\begin{align*}
    \mathrm{d}l&=\mathrm{tr}(\frac{\partial l}{\partial \mA_{3(P)}}^{\top} \mP_R^{\top} \mathrm{d}(\mA_{2(P)}) \mP_C^{\top})\\
    &=\mathrm{tr}(\mP_C^{\top}\frac{\partial l}{\partial \mA_{3(P)}}^{\top} \mP_R^{\top} \mathrm{d}\mA_{2(P)})\\
    &=\mathrm{tr}((\mP_R \frac{\partial l}{\partial \mA_{3(P)}} \mP_C)^{\top}\mathrm{d}\mA_{2(P)}),
\end{align*}
that is:
\begin{equation}
    \frac{\partial l}{\partial \mA_{2(P)}}=\mP_R\frac{\partial l}{\partial \mA_{3(P)}}\mP_C = \mP_R\frac{\partial l}{\partial \mA_{2}}\mP_C 
\end{equation}
by Eq.~\ref{a2a3}.

With $\mH_{(P)}=\mP_R\mH\mP_C^{\top}$ and Eq.~\ref{w2}, the gradient:
\begin{align*}
    \frac{\partial l}{\partial \mW_{2(P)}}&=\frac{\partial l}{\partial \mA_{2(P)}}^{\top}\mH_{(P)}\\
    &=\mP_C^{\top}\frac{\partial l}{\partial \mA_{2}}^{\top}\mP_R^{\top}\cdot \mP_R \mH \\
    &=\mP_C^{\top}\frac{\partial l}{\partial \mA_{2}}^{\top}\mH  \\
    &=\mP_C^{\top} \frac{\partial l}{\partial \mW_{2}} ,
\end{align*}
that is:
\begin{equation}\label{w2p}
    \frac{\partial l}{\partial \mW_{2(P)}}=\mP_C^{\top} \frac{\partial l}{\partial \mW_{2}} .
\end{equation}

By Eq.~\ref{a1} and Eq.~\ref{w2p}, we have
\begin{align*}
    \frac{\partial l}{\partial \mA_{1(P)}}&=\frac{\partial l}{\partial \mA_{2(P)}} \mW_{2(P)} \odot a'(\mA_{1(P)})\\
    &=[\mP_R \frac{\partial l}{\partial \mA_{2}}\mP_C \cdot \mP_C^{\top}\mW_2]\odot [\mP_Ra'(\mA_{1})]\\
    &=[\mP_R \frac{\partial l}{\partial \mA_{2}}\mW_2]\odot [\mP_Ra'(\mA_{1})]\\
    &=\mP_R[\frac{\partial l}{\partial \mA_{2}} \mW_{2} \odot a'(\mA_{1})]\\
    &=\mP_R\frac{\partial l}{\partial \mA_{1}},
\end{align*}
that is:
\begin{equation}
    \frac{\partial l}{\partial \mA_{1(P)}}= \mP_R\frac{\partial l}{\partial \mA_{1}}.
\end{equation}

The weight $\mW_{1(P)}$ in the MLP has the following gradient by Eq.~\ref{W1}:
\begin{align*}
    \frac{\partial l}{\partial \mW_{1(P)}}&=\frac{\partial l}{\partial \mA_{1(P)}}^{\top}\mA_{(P)}\\
    &=\frac{\partial l}{\partial \mA_{1}}^{\top}\mP_R^{\top}\cdot \mP_R\mA \mP_C\\
    &=\frac{\partial l}{\partial \mW_{1}}\mP_C,
\end{align*}
that is:
\begin{equation}
    \frac{\partial l}{\partial \mW_{1(P)}}=\frac{\partial l}{\partial \mW_{1}}\mP_C.
\end{equation}

And we come to the attention operation, from Eq.~\ref{A}, we have
\begin{align*}
    \frac{\partial l}{\partial \mA_{(P)}}&=\frac{\partial l}{\partial \mA_{1(P)}}\mW_{1(P)}\\
    &=\mP_R\frac{\partial l}{\partial \mA_{1}} \mW_1 \mP_C\\
    &=\mP_R\frac{\partial l}{\partial \mA}\mP_C,
\end{align*}
that is:
\begin{equation}
    \frac{\partial l}{\partial \mA_{(P)}}=\mP_R\frac{\partial l}{\partial \mA}\mP_C.
\end{equation}

Hence we observe the permutation rules for the gradients of the intermediate-layer outputs vary from the gradients of the weights. As for the gradients of the softmax-layer output, we have
\begin{align*}
    \frac{\partial l}{\partial \mS_{(P)}}
    &=\frac{\partial l}{\partial \mA_{(P)}}\mV_{(P)}^{\top}\\
    &=\mP_R\frac{\partial l}{\partial \mA}\mP_C \cdot \mP_C^{\top}\mV^{\top}\mP_R^{\top}\\
    &=\mP_R\frac{\partial l}{\partial \mA}\mV^{\top}\mP_R^{\top}\\
    &=\mP_R\frac{\partial l}{\partial \mS}\mP_R^{\top},
\end{align*}
that is:
\begin{equation} \label{sbackward}
    \frac{\partial l}{\partial \mS_{(P)}}=\mP_R\frac{\partial l}{\partial \mS}\mP_R^{\top}.
\end{equation}

Since $\mS_{(P)}$ follows Eq.~\ref{sforward}, we have the gradients for $\mQ_{(P)}$ combining with Eq.~\ref{sbackward}:
\begin{align*}
    \frac{\partial l}{\partial \mQ_{(P)}}
    &=\frac{1}{\sqrt{d}}[(diag(\mS_{(P)})-\mS_{(P)}^{\top}\mS_{(P)})\frac{\partial l}{\partial \mS_{(P)}}]\mK_{(P)}\\
    &=\frac{1}{\sqrt{d}}[(\mP_Rdiag(\mS)\mP_R^{\top}-\mP_R\mS^{\top}\mP_R^{\top}\cdot \mP_R\mS\mP_R^{\top})\\ &~~~~~~~~\cdot\mP_R\frac{\partial l}{\partial \mS}\mP_R^{\top}]\mP_R\mK\mP_C^{\top}\\
    &=\frac{1}{\sqrt{d}}[(\mP_Rdiag(\mS)\mP_R^{\top}-\mP_R\mS^{\top}\mS\mP_R^{\top})\\ &~~~~~~~~\cdot\mP_R\frac{\partial l}{\partial \mS}\mP_R^{\top}]\mP_R\mK\mP_C^{\top}\\
    &=\frac{1}{\sqrt{d}}[\mP_R(diag(\mS)-\mS^{\top}\mS)\\ &~~~~~~~~\cdot\mP_R^{\top}\cdot \mP_R\frac{\partial l}{\partial \mS}\mP_R^{\top}]\mP_R\mK\mP_C\\
    &=\frac{1}{\sqrt{d}}[\mP_R(diag(\mS)-\mS^{\top}\mS)\frac{\partial l}{\partial \mS}\mP_R^{\top}]\mP_R\mK\mP_C\\
    &=\frac{1}{\sqrt{d}}\mP_R(diag(\mS)-\mS^{\top}\mS)\frac{\partial l}{\partial \mS}\mP_R^{\top}\cdot \mP_R\mK\mP_C\\
    &=\mP_R\frac{1}{\sqrt{d}}(diag(\mS)-\mS^{\top}\mS)\frac{\partial l}{\partial \mS}\mK\mP_C\\
    &=\mP_R\frac{\partial l}{\partial \mQ}\mP_C.
\end{align*}
By a similar derivation on $\mK$ we obtain:
\begin{equation}
    \frac{\partial l}{\partial \mK_{(\mP)}}=\mP_R\frac{\partial l}{\partial \mK}\mP_C.
\end{equation}

Following a similar proof to the gradients of $\mW_{1(P)}$ or $\mW_{2(P)}$, we could easily derive:
\begin{align}
    \frac{\partial l}{\partial \mW_{Q(P)}}&=\mP_C^{\top}\frac{\partial l}{\partial \mW_Q}\mP_C,\\
    \frac{\partial l}{\partial \mW_{K(P)}}&=\mP_C^{\top}\frac{\partial l}{\partial \mW_K}\mP_C.
\end{align}

By Eq.~\ref{V}, the gradient of $\mV_{(P)}$ is
\begin{align*}
    \frac{\partial l}{\partial \mV_{(P)}}
    &=\mS_{(P)}^{\top}\frac{\partial l}{\partial \mA_{(P)}}\\
    &=\mP_R \mS \mP_R^{\top}\cdot \mP_R \frac{\partial l}{\partial \mA}\mP_C\\
    &=\mP_R\frac{\partial l}{\partial \mV}\mP_C,
\end{align*}
and thus we have
\begin{align}
    \frac{\partial l}{\partial \mV_{(P)}}&=\mP_R\frac{\partial l}{\partial \mV}\mP_C, \label{eq:veq}\\
    \frac{\partial l}{\partial \mW_{V(P)}}&=\mP_C^{\top}\frac{\partial l}{\partial \mW_V}\mP_C.
\end{align}

So far, we have proved the rule for the gradient of weight matrices:
\begin{equation} \label{relat}
    \frac{\partial l}{\partial \mW_{i(P)}}=\mP_C^{\top}\frac{\partial l}{\partial \mW_i}\mP_C,~~i \in \{ Q, K, V \}.
\end{equation}
\begin{equation}\label{eq:relat_MLP}
    \frac{\partial l}{\partial \mW_{1(P)}}=\frac{\partial l}{\partial \mW_{1}}\mP_C, ~~\frac{\partial l}{\partial \mW_{2(P)}}=\mP_C^{\top} \frac{\partial l}{\partial \mW_{2}} .
\end{equation}

$\mW_{i(P)}$ are the weights of $\mathrm{Enc}_{(P)}$ while $\mW_i$ are the weights of $\mathrm{Enc}$. By induction, we can reach the conclusion that if a Transformer encoder block is randomly initialized and trained with $\mZ_{(P)}$, it would eventually learn to become $\mathrm{Enc}_{(P)}$, the weights of which are associated with $\mathrm{Enc}$ by Eq.~\ref{relat} and Eq.~\ref{eq:relat_MLP}. The proof of backward equivariance on the linear projection in the attention is omitted as its proof is similar and the conclusion is the same with Eq.~\ref{relat}. Hence we have proved backward permutation in-/equi-variance.
\end{proof}

\subsection{Proofs on Embeddings}\label{proof_edge}
We show in this section that the parameters of the embedding layer $F_1$, including the position embeddings, are the same despite Alg.~\ref{alg1} is applied or not.
\begin{theorem}
The parameters of $F_1$ trained with or without permutation are the same.
\end{theorem}

\begin{proof}
We denote the output of $F_1$ as $\mZ_0$, and the input of the Transformer backbone as $\mZ$. In normal setting (Eq.~\ref{eq:training_scenario}), $\mZ$ equals to $\mZ_0$, and so do their gradients. In the permuted setting where we use subscript $_{(P)}$ to denote all the varibles, the input to the backbone is the permuted output of $F_{1(P)}$:
\begin{equation}
    \mZ_{(P)} = \mP_R \mZ_{0(P)} \mP_C.
\end{equation}

To prove the weights of $F_{1(P)}$ is equivalent to those of $F_1$, we need to prove:
\begin{equation}
    \frac{\partial l}{\partial \mZ_{0(P)}} = \frac{\partial l}{\partial \mZ_0}.
\end{equation}
It is clear that
\begin{align*}
    \mathrm{d}l &= \mathrm{tr}(\frac{\partial l}{\partial \mZ_{(P)}}^{\top}\mathrm{d}\mZ_{(P)}) \\
                &=  \mathrm{tr}(\mP_C^{\top}\frac{\partial l}{\partial \mZ}^{\top}\mP_R^{\top}~ \mP_R\mathrm{d}\mZ_{0(P)}\mP_C)\\
                &=  \mathrm{tr}(\mP_C~\mP_C^{\top}\frac{\partial l}{\partial \mZ}^{\top}\mathrm{d}\mZ_{0(P)})\\
                &=  \mathrm{tr}(\frac{\partial l}{\partial \mZ}^{\top}\mathrm{d}\mZ_{0(P)}),
\end{align*}
where the second equality holds by Lemma~\ref{lem:gradient_feature_equivalent}. Hence,
\begin{equation}\label{eq:zeq}
    \frac{\partial l}{\partial \mZ_{0(P)}}=\frac{\partial l}{\partial \mZ}=\frac{\partial l}{\partial \mZ_0}.
\end{equation}
The invariance of the weights of $F_1$ can be derived from Eq.~\ref{eq:zeq}. 
\end{proof}

\blue{It is worth noting that the position embeddings are added before the permutation, so Transformer gets the right position information, since Transformer modeling the position by the position embeddings instead of the order of the input.}

\subsection{Proofs on Masked Attention}\label{proof_mask}
Masked attention is an essential component in the Transformer Decoder Block which is the backbone of the generetive language model \cite{GPT2, GPT3}. The token permutation can hardly pass through the nonlinear effect of the masked attention, but the column permutation cancels out before the mask takes effect. Thus for the masked attention, we apply column permutations only.
\begin{theorem}
    The results of Masked Softmax with or without column permutation are the same:
    \begin{equation}
        MS_{(P)} = MS.
    \end{equation}
\end{theorem}
\begin{proof}
    For $MS=Softmax(\frac{\mQ\mK^{\top}}{\sqrt{d}})$:
\begin{align}
    MS_{(P)}&=Masked-Softmax(\frac{\mQ_{(P)}\mK_{(P)}^{\top}}{\sqrt{d}})\\
    &=Masked-Softmax(\frac{\mQ\mP_C\cdot \mP_C^{\top}\mK^{\top}}{\sqrt{d}})\\
    &=Masked-Softmax(\frac{\mQ\mK^{\top}}{\sqrt{d}})\\
    &=Masked-Softmax(\frac{\mQ\mK^{\top}}{\sqrt{d}})\\&=MS. \label{msforward}
\end{align}
\end{proof}

\section{Detailed Experimental Setup}
\label{sec:exp_setup}
\subsection{Training Setup}
In fine-tuning ViT-Base on Cifar10, an Adam optimizer with a fixed learning rate $10^{-4}$ is used. The model is trained for 5 epochs with a cross-entropy loss. In fine-tuning Bert and GPT2 on IMDB classification, an Adam optimizer with a fixed learning rate $10^{-5}$ is used. The models are trained for 2 epochs.

For the `unauthorized' and train-from-scratch setting on Cifar10 of ViT, the default training setting of ViT is used. Random data augmentation and cosine scheduler are added and the model is trained for 200 epochs till convergence.

The text generation in the properties validation experiments is zero-shot learning on huggingface pre-trained GPT2. The model is fine-tuned with Adam optimizer and learning rate $10^{-5}$ for 1 epoch.

For the CelebA attribute classification task in the `privacy-preserving split learning' experiments, we adopt {\tt timm} pre-trained model vit\_base\_patch16\_224 (ViT-Base) to transfer to a 40-binary-attribute classification task. A SGD optimizer is used with a cosine scheduler, for which the (initial, final) learning rate are set to $(0.05, 2 \times 10^{-4})$ and $(5\times 10^{-4}, 2 \times 10^{-6})$ for the classification head and the encoder blocks, respectively.

\subsection{Threat Model of Privacy-Preserving Split Learning}
The threat model is consist with \cite{patchshuffling,jeong2022privacy}. We assume the edge holds a private training data set $\sD_{train} = \{X,Y\}$, where $X$ are the private data and $Y$ are the private labels. The edge aims at training a model with the assistance of the cloud, yet without exposing the private data. The cloud possesses powerful computing power and is honest-but-curious, meaning it obeys the protocol and performs the learning task accordingly, but is curious about the private data. The edge selects a model and splits it into three parts: $F_1, \mathrm{Enc}, F_2.$ $F_1, F_2$ are parts close to the input layer and the output layer, respectively, which are sufficiently lightweight to deploy on the edge, whereas $\mathrm{Enc}$ is the major part run in the cloud.

Referring to the loss function as $L_{task}$ and the local privacy-preserving method as $M$, the ultimate goal is to jointly train $F_1, \mathrm{Enc}, F_2$ to
\begin{equation}\label{vanilla_split_learning}
	\underset{F_1,\mathrm{Enc},F_2}{\textrm{minimize}} ~~L_{task}(F_2(\mathrm{Enc}(F_1(X))),Y),
\end{equation}
without the edge revealing $X, Y$ to the cloud. Accessing $F_1(X)$ should not permit the cloud to infer about $X$. The cloud could launch black-box inversion attacks: 

\textbf{Black-box attackers} collect the auxiliary data set $X_{aux}$ and the corresponding features under protection mechanism $M$ as $M(F_1(X_{aux}))$, maybe over multiple training rounds. The attacker trains an inversion model $G$ over $(X_{aux}, M(F_1(X_{aux})))$ to invert the raw input from features \citep{dosovitskiy2016inverting,fredrikson2015model,melis2019exploiting} by
\begin{equation}\label{blackbox}
	\underset{G}{\textrm{minimize}} ~~ L_{atk}(G(M(F_1(X_{aux}))),X_{aux}).
\end{equation}
The loss $L_{atk}$ can be the mean square error (MSE) between the reconstructed input $\tilde{X}_{aux}$ and $X_{aux}$. At convergence, $G$ works as a decoder to invert features into inputs. We adopt the MAE decoder $G$ with base-size Transformer backbone and a Tanh activation layer, pre-trained on ImageNet, as the model inversion model. We train $G$ with an AdamW optimizer with a learning rate of $10^{-4}$ for 30 epochs.

\end{document}